\documentclass[11pt]{article}

\usepackage{fullpage,amsmath,amssymb,orcidlink,hyperref,url,supertabular}
 
\newtheorem{Corollary}{Corollary}
\newtheorem{Remark}{Remark}
\newtheorem{Definition}{Definition}
\newtheorem{Proposition}{Proposition}
\newenvironment{proof}{\paragraph{Proof:}}{\hfill$\square$}

\def\BC{\mathrm{BC}}
\def\JSD{\mathrm{JSD}}
\def\tD{{\tilde{D}}}

\def\erf{\mathrm{erf}}
\def\dt{\mathrm{d}t}
\def\TV{\mathrm{TV}}
\def\bbR{\mathbb{R}}
\def\eqdef{{:=}}
\def\KL{{\mathrm{KL}}}
\def\dmu{{\mathrm{d}\mu}}
\def\tm{{\tilde{m}}}
\def\JS{\mathrm{JS}}
\def\calX{\mathcal{X}}
\def\calE{\mathcal{E}}

\def\calN{\mathcal{N}}
\def\tr{\mathrm{tr}}
\def\inner#1#2{\left\langle #1,#2\right\rangle}
\def\PD{\mathrm{PD}}
\def\st{{\ :\ }}
\def\tM{{\tilde{M}}}
\def\eps{\epsilon}

\def\tM{{\tilde{M}}}
\def\tG{{\tilde{G}}}
\def\calG{\mathcal{G}}

\title{Two tales for a geometric Jensen--Shannon divergence\footnote{This work appeared in~\cite{nielsen2025two}.}}

\author{Frank Nielsen\ \orcidlink{0000-0001-5728-0726}\\ \ \\ Sony Computer Science Laboratories Inc., Japan}

\date{} 

\begin{document}
\maketitle

\sloppy

\begin{abstract}
The geometric Jensen--Shannon divergence (G-JSD) gained popularity in machine learning and information sciences thanks to its closed-form expression between Gaussian distributions. In this work, we introduce an alternative definition of the geometric Jensen--Shannon divergence tailored to positive densities which does not normalize geometric mixtures. 
This novel divergence is termed the extended G-JSD as it applies to the more general  case of positive measures.
We report explicitly the gap between the extended G-JSD  and the G-JSD when considering probability densities, 
and show how to express the G-JSD and extended G-JSD using the Jeffreys divergence and the Bhattacharyya distance or Bhattacharyya coefficient.
The extended G-JSD is proven to be a $f$-divergence which is a separable divergence satisfying information monotonicity and invariance in information geometry.  
We derive corresponding closed-form formula for the two types of G-JSDs when considering the case of multivariate Gaussian distributions often met in applications.
We consider Monte Carlo stochastic estimations and approximations of the two types of G-JSD using the projective $\gamma$-divergences.
Although the square root of the JSD yields a metric distance, we show that this is not anymore the case for the two types of G-JSD.
Finally, we explain how these two types of geometric JSDs can be interpreted as regularizations of the ordinary JSD.
\end{abstract}

\noindent {\bf Keywords}: Jensen--Shannon divergence; quasi-arithmetic means; total variation distance; Bhattacharyya distance; Chernoff information; Jeffreys divergence; Taneja divergence;  geometric mixtures; exponential families; projective $\gamma$-divergences; $f$-divergence; separable divergence; information monotonicity.

\section{Introduction}

\subsection{Kullback--Leibler and Jensen--Shannon divergences}

Let $(\calX,\calE,\mu)$ be a measure space on the sample space $\calX$, $\sigma$-algebra of events $\calE$, with $\mu$ a prescribed positive measure on the measurable space $(\calX,\calE)$ (e.g., counting measure or Lebesgue measure). Let $M_+(\calX)=\{Q\}$ be the set of positive distributions $Q$ and $M_+^1(\calX)=\{P\}$ be the subset of probability measures $P$. We denote by  $M_\mu=\{\frac{\mathrm{d}Q}{\dmu} \st Q\in M_+(\calX)\}$ and 
$M^1_\mu=\{\frac{\mathrm{d}P}{\dmu} \st P\in M_+^1(\calX)\}$ the corresponding sets of Radon-Nikodym positive and probability densities, respectively.

Consider two probability measures $P_1$ and $P_2$ of $M_+^1(\calX)$ with Radon-Nikodym densities with respect to $\mu$
 $p_1\eqdef\frac{\mathrm{d}P_1}{\dmu}\in M^1_\mu$ and $p_2\eqdef\frac{\mathrm{d}P_2}{\dmu}\in M_\mu^1$, respectively.
The deviation of $P_1$ to $P_2$ (also called distortion, dissimilarity, or deviance) is commonly measured in information theory~\cite{cover1999elements} by the Kullback--Leibler divergence (KLD): 
\begin{equation}\label{eq:kld}
\KL(p_1,p_2) \eqdef \int p_1\log\frac{p_1}{p_2}\,\dmu = E_{p_1}\left[\log\frac{p_1}{p_2}\right].
\end{equation}

Informally, the KLD quantifies the information lost when $p_2$ is used to approximate $p_1$ by measuring on average the surprise when outcomes sampled from $p_1$ are assumed to emanate from $p_2$:
Shannon entropy $H(p)=\int p\log\frac{1}{p}\,\dmu$ is the expected surprise $H(p)=E_p[-\log p]$ where $-\log p(x)$ measures the surprise of the outcome $x$. Logarithms are taken to base $2$ when information is measured in bits, and to base $e$ when it is measured in nats.
 Gibbs' inequality assert that $\KL(P_1,P_2)\geq 0$ with equality if and only if $P_1=P_2$ $\mu$-almost everywhere.
Since $\KL(p_1,p_2)\not=\KL(p_2,p_1)$, various symmetrization schemes of the KLD have been proposed in the literature~\cite{cover1999elements} (e.g., Jeffreys divergence~\cite{cover1999elements,jeffreys1998theory}, resistor average divergence~\cite{johnson2001symmetrizing} (harmonic KLD symmetrization), Chernoff information~\cite{cover1999elements},   etc.)

An important symmetrization technique of the KLD is the Jensen--Shannon divergence~\cite{JSD-1991,fuglede2004jensen} (JSD):
\begin{equation}\label{eq:jsd}
\JS(p_1,p_2)\eqdef \frac{1}{2}\,\left(\KL(p_1,a) + \KL(p_2,a) \right),
\end{equation}
where $a=\frac{1}{2}p_1+\frac{1}{2}p_2$ denotes the statistical mixture of $p_1$ and $p_2$.
The JSD is guaranteed to be upper bounded by $\log 2$ even when the support of $p_1$ and $p_2$ differ, making it attractive in applications.
Furthermore, its square root $\sqrt{\JS}$ yields a metric distance~\cite{endres2003new,okamura2023metrization}.

The JSD can be extended to a set of densities to measure the  diversity of the set as an information radius~\cite{sibson1969information}.
In information theory, the JSD can also be interpreted as an information gain~\cite{endres2003new} since it can be equivalently written as
$$
\JS(p_1,p_2)= H\left(\frac{1}{2}p_1+\frac{1}{2}p_2\right)-\frac{H(p_1)+H(p_2)}{2},
$$
where $H(p)=-\int p\log p\, \dmu$ is Shannon entropy (Shannon entropy for discrete measures  and differential entropy for continuous measures).
The JSD was also defined in the setting of quantum information~\cite{briet2009properties} where it was also proven that its square root yields a metric distance~\cite{virosztek2021metric}.

\begin{Remark}
Both the KLD and the JSD belong to the family of $f$-divergences~\cite{AliSilvey-1966,Csiszar-1967} defined for a convex generator $f(u)$ (strictly convex at $1$) by:
$$
I_f(p_1,p_2) \eqdef \int p_1\, f\left(\frac{p_2}{p_1}\right) \,\dmu.
$$
Indeed, we have $\KL(p_1,p_2)=I_{f_\KL}(p_1,p_2)$ and $\JS(p_1,p_2)=I_{f_\JS}(p_1,p_2)$ for the following generators:
\begin{eqnarray*}
f_\KL(u) &\eqdef& -\log u,\\
f_\JS(u) &\eqdef& -(1+u)\log\frac{1+u}{2}+u\log u.
\end{eqnarray*}
The family of $f$-divergences are the invariant divergences in information geometry~\cite{IG-2016,nielsen2020elementary,nielsen2022many}.
The $f$-divergences guarantee  information monotonicity by coarse graining~\cite{IG-2016} 
(also called lumping in information theory~\cite{csiszar2004information}). 
Using Jensen inequality, we get that $I_f(p_1,p_2)\geq f(1)$.
\end{Remark}

\begin{Remark}
The metrization of $f$-divergences was studied in~\cite{osterreicher2003new}.
Once a metric distance $D(p_1,p_2)$ is given, we may use the following metric transform~\cite{schoenberg1938metric} to obtain another metric which is guaranteed to be bounded by $1$:
$$
0\leq d(p_1,p_2)=\frac{D(p_1,p_2)}{1+D(p_1,p_2)}\leq 1.
$$
\end{Remark}

\subsection{Jensen--Shannon symmetrization of dissimilarities with generalized mixtures}

In~\cite{GJSD-2019}, a generalization of the KLD Jensen--Shannon symmetrization scheme~\cite{nielsen2020generalization} was studied for arbitrary statistical dissimilarity $D(\cdot,\cdot)$ by using an arbitrary weighted mean~\cite{bullen2013handbook} $M_\alpha$.
A generic weighted mean $M_\alpha(a,b)=M_{1-\alpha}(b,a)$ for $a,b\in\bbR_{>0}$ is a continuous symmetric monotonic map $\alpha\in [0,1]\mapsto M_\alpha(a,b)$ such that $M_0(a,b)=b$ and $M_1(a,b)=1$.
For example, the quasi-arithmetic means~\cite{bullen2013handbook} are defined according to a monotonous continuous function $\phi$ as follows:
$$
M_\alpha^\phi(a,b)\eqdef \phi^{-1}\left(\alpha\phi(a)+(1-\alpha)\phi(b)\right).
$$ 
When $\phi_p(u)=u^p$, we get the $p$-power mean $M_\alpha^{\phi_p}(a,b)=(\alpha a^p+(1-\alpha) b^p)^{\frac{1}{p}}$ for $p\in\bbR\backslash\{0\}$.
We extend $\phi_p$ for $p=0$ by defining $\phi_0(u)=\log u$, and get $M_\alpha^{\phi_0}(a,b)=a^\alpha b^{1-\alpha}$, the weighted geometric mean $G_\alpha$.

Let us recall the generalization of the Jensen--Shannon symmetrization scheme of a dissimilarity measure presented in~\cite{GJSD-2019}:

\begin{Definition}[$(\alpha,\beta)$ M-JS dissimilarity~\cite{GJSD-2019}]\label{def:MJSD}
The Jensen--Shannon skew symmetrization of a statistical dissimilarity $D(\cdot,\cdot)$ with respect to an arbitrary weighted bivariate mean $M_\alpha(\cdot,\cdot)$ is given by:

\begin{equation}\label{eq:MJSD}
D^\JS_{M_\alpha,\beta}(p_1,p_2) \eqdef \beta \, D\left(p_1,{(p_1p_2)}_{M_\alpha}\right) + (1-\beta)\, D\left(p_2,{(p_1p_2)}_{M_\alpha}\right),\quad 
(\alpha,\beta) \in (0,1)^2,  
\end{equation}

where ${(p_1p_2)}_{M_\alpha}$ is the statistical normalized weighted  $M$-mixture of $p_1$ and $p_2$:

\begin{equation}
{(p_1p_2)}_{M_\alpha}(x) \eqdef \frac{M_\alpha(p_1(x),p_2(x))}{\int M_\alpha(p_1(x),p_2(x))\,\dmu(x)}.
\end{equation}
\end{Definition}

\begin{Remark}
A more general definition is given in~\cite{GJSD-2019} by using another arbitrary weighted mean $N_\beta$ to average the two dissimilarities in Eq.~\ref{eq:MJSD}:
\begin{equation}\label{eq:MNJSD}
D^\JS_{M_\alpha,N_\beta}(p_1,p_2) \eqdef  N_\beta\left( D\left(p_1,{(p_1p_2)}_{M_\alpha}\right) , D\left(p_2,{(p_1p_2)}_{M_\alpha}\right)\right),\quad 
(\alpha,\beta) \in (0,1)^2. 
\end{equation}
When $N_\beta=A_\alpha$ the weighted arithmetic mean $A_\alpha(a,b)=\alpha a+(1-\alpha)b$, Eq.~\ref{eq:MNJSD} amounts to Eq.~\ref{eq:MJSD}.
\end{Remark}

When $\alpha=\frac{1}{2}$, we   write for short ${(p_1p_2)}_M$ instead of ${(p_1p_2)}_{M_{\frac{1}{2}}}$   in the reminder.

When $D=\KL$, $M=N=A_{\frac{1}{2}}$, Eq.~\ref{eq:MNJSD} yields the Jensen--Shannon divergence of Eq.~\ref{eq:jsd}:
$\JS(p_1,p_2)=\KL^\JS_{A_{\frac{1}{2}},A_{\frac{1}{2}}}(p_1,p_2)=\KL^\JS_{A,A}(p_1,p_2)$.

Lower and upper bounds for the skewed $\alpha$-Jensen--Shannon divergence  were reported in~\cite{yamano2019some}.

The abstract mixture normalizer of  ${(p_1p_2)}_{M_\alpha}$ shall be denoted by 
$$
Z_{M_\alpha}(p_1,p_2)  \eqdef \int M_\alpha(p_1(x),p_2(x))\,\dmu(x),
$$
 so that the normalized $M$-mixture is written as 
${(p_1p_2)}_{M_\alpha}(x)=\frac{M_\alpha(p_1(x),p_2(x))}{Z_{M_\alpha}(p_1,p_2)}$.
The normalizer $Z_{M_\alpha}(p_1,p_2)$ is always finite and thus the weighted $M$-mixtures ${(p_1p_2)}_{M_\alpha}$ are well-defined:

\begin{Proposition}
For any generic weighted mean $M_\alpha$, we have the normalizer of the weighted $M$-mixture bounded by $2$:
$$
0\leq Z_{M_\alpha}(p_1,p_2)\leq 2.
$$
\end{Proposition}

\begin{proof}
Since $M_\alpha$ is a scalar weighted mean, it satisfies the following in-betweenness property: 
\begin{equation}
\min\{p_1(x),p_2(x)\} \leq M_\alpha(p_1(x),p_2(x)) \leq \max\{p_1(x),p_2(x)\}.
\end{equation}
Hence, by using the following two identities for $a\geq 0$ and $b\geq 0$: 
\begin{eqnarray*}
\min\{a,b\} &=& \frac{a+b}{2}-\frac{1}{2}|a-b|,\\
\max\{a,b\} &=& \frac{a+b}{2}+\frac{1}{2}|a-b|,
\end{eqnarray*}
 we  get
\begin{eqnarray}
\int \min\{p_1(x),p_2(x)\}\, \dmu(x) \leq &\int M_\alpha(p_1(x),p_2(x))\, \dmu(x)& \leq \int \max\{p_1(x),p_2(x)\}\, \dmu(x),\nonumber\\
0\leq 1-\TV(p_1,p_2) \leq & Z_{M_\alpha}(p_1,p_2) & \leq 1+\TV(p_1,p_2)\leq 2,\label{eq:zub}
\end{eqnarray}
where 
$$
\TV(p_1,p_2)\eqdef\frac{1}{2}\int |p_1-p_2|\,\dmu,
$$ 
is the total variation distance, upper bounded by $1$.
When the support of the densities $p_1$ and $p_2$ intersect (i.e., non-singular probability measures $P_1$ and $P_2$), we have $Z_{M_\alpha}(p_1,p_2)>0$ 
and therefore the weighted $M$-mixtures $(p_1p_2)_{M_\alpha}$ are well-defined.
\end{proof}

The generic Jensen--Shannon symmetrization of dissimilarities given in Definition~\ref{def:MJSD} allows us to re-interpret some well-known statistical dissimilarities:

For example, the Chernoff information~\cite{cover1999elements,nielsen2022revisiting} is defined by 
\begin{equation}\label{eq:ci}
C(p_1,p_2)\eqdef \max_{\alpha\in(0,1)} B_\alpha(p_1,p_2),
\end{equation}
 where $B_\alpha(p_1,p_2)$ denotes the $\alpha$-skewed Bhattacharrya distance:
\begin{equation}
B_\alpha(p_1,p_2) \eqdef -\log\int p_1^\alpha\, p_2^{1-\alpha}\,\dmu
\end{equation} 
When $\alpha=\frac{1}{2}$, we note $B(p_1,p_2)=B_{\frac{1}{2}}(p_1,p_2)$ the Bhattacharrya distance.
Notice that the Bhattacharrya distance is not a metric distance as it violates the triangle inequality of metrics.

Using the framework of JS-symmetrization of dissimilarities, we can reinterpret the Chernoff information as
$$
C(p_1,p_2)=({\KL^*})^\JS_{G_{\alpha^*},A_{\frac{1}{2}}}(p_1,p_2),
$$
where $\alpha^*$ is provably the unique optimal skewing factor in Eq.~\ref{eq:ci} such that we have~\cite{nielsen2022revisiting}: 
\begin{eqnarray*}
C(p_1,p_2) &=& \KL^*(p_1,(p_1p_2)_{G_{\alpha^*}})=\KL^*(p_2,(p_1p_2)_{G_{\alpha^*}}),\\
&=& \frac{1}{2}\,\left( \KL^*(p_1,(p_1p_2)_{G_{\alpha^*}}) + \KL^*(p_2,(p_1p_2)_{G_{\alpha^*}}) \right),
\end{eqnarray*}
where ${\KL}^*$ denotes the reverse KLD: 
$$
\KL^*(p_1,p_2)\eqdef \KL(p_2,p_1).
$$
Note that KLD is sometimes called the forward KLD (e.g.,\cite{jerfel2021variational}), and we have ${\KL^*}^*(p_1,p_2)=\KL(p_1,p_2)$.

Although arithmetic mixtures are most often used in Statistics, the geometric mixtures are also encountered, like for example in Bayesian statistics~\cite{asadi2018mixture}, or in Markov chain Monte Carlo annealing~\cite{grosse2013annealing}, just to give two examples. 
In information geometry, statistical power mixtures based on the homogeneous power means are used to perform stochastic integration of statistical models~\cite{amari2007integration}.

\begin{Proposition}[Bhattacharyya distance as G-JSD]\label{prop:JSBhat}
The Bhattacharyya distance~\cite{bhattacharyya1946measure} and the $\alpha$-skewed Bhattacharyya distances can be interpreted as  JS-symmetrizations of the reverse KLD with respect to the geometric mean $G$:
\begin{eqnarray*}
B(p_1,p_2) &\eqdef& -\log\int \sqrt{p_1\, p_2} \,\dmu=(\KL^*)^\JS_{G}(p_1,p_2),\\
B_\alpha(p_1,p_2) &\eqdef& -\log\int p_1^\alpha\, p_2^{1-\alpha} \,\dmu=(\KL^*)^\JS_{G_\alpha}(p_1,p_2).
\end{eqnarray*}
\end{Proposition}

\begin{proof}
Let $m=(p_1p_2)_G=\frac{\sqrt{p_1p_2}}{Z(p_1,p_2)}$ denote the weighted geometric mixture 
with normalizer $Z_G(p_1,p_2)=\int \sqrt{p_1p_2}\,\dmu$.
By definition of the JS-symmetrization of the reverse KLD, we have
\begin{eqnarray*}
({\KL^*})_G^\JS(p_1,p_2) &\eqdef& \frac{1}{2}\left(  {\KL^*}(p_1,(p_1p_2)_G) +  {\KL^*}(p_2,(p_1p_2)_G) \right),\\
&=& \frac{1}{2}\left(  {\KL}((p_1p_2)_G,p_1) +  {\KL}((p_1p_2)_G,p_2) \right),\\
&=& \frac{1}{2}\left(\int \left(  m \,\log \frac{\sqrt{p_1p_2} }{p_1\, Z_G(p_1,p_2)}  
+ m\, \log \frac{\sqrt{p_1p_2}}{p_2\, Z_G(p_1,p_2)} \right)\,\dmu
 \right),\\
&=&  \frac{1}{2}\left(\int \frac{1}{2} m\log \frac{p_2}{p_1}\,\frac{p_1}{p_2} \dmu   -2\log Z_G(p_1,p_2) \int m\,\dmu \right),\\
&=& -\log Z_G(p_1,p_2)=:B(p_1,p_2).
\end{eqnarray*}

The proof carries on similarly for the $\alpha$-skewed JS-symmetrization of the reverse KLD:
We now let $m_\alpha=(p_1p_2)_{G_\alpha}=\frac{p_1^\alpha p_2^{1-\alpha}}{Z_{G_\alpha}(p_1,p_2)}$ be the $\alpha$-weighted geometric mixture 
with normalizer $Z_{G_\alpha}(p_1,p_2)=\int p_1^\alpha p_2^{1-\alpha}\,\dmu$, written as $Z_{G_\alpha}$ for short below:

\begin{eqnarray*}
{\KL^*}_{G_\alpha,\alpha}^\JS(p_1,p_2) &\eqdef&  \alpha\,  {\KL^*}(p_1,(p_1p_2)_{G_\alpha}) + (1-\alpha)\, {\KL^*}(p_2,(p_1p_2)_{G_\alpha}),\\
&=&   \alpha\,  {\KL}(m_\alpha, p_1) + (1-\alpha)\, {\KL}(m_\alpha,p_2) ,\\
&=&  \int\left( \alpha m_\alpha\log\frac{p_1^\alpha\, p_2^{1-\alpha}}{Z_{G_\alpha}\, p_1} 
+(1-\alpha) m_\alpha\log\frac{p_1^\alpha p_2^{1-\alpha}}{Z_{G_\alpha}\, p_2} 
 \right)\dmu,\\
&=&  -(\alpha+1-\alpha)\log Z_{G_\alpha} \int m_\alpha\,\dmu + \int m_\alpha \log \left(\frac{p_2}{p_1}\right)^{\alpha(1-\alpha)}\, \left(\frac{p_1}{p_2}\right)^{\alpha(1-\alpha)} \dmu,\\
&=& -\log Z_{G_\alpha}(p_1,p_2)=:{B_\alpha}(p_1,p_2).
\end{eqnarray*}
\end{proof}

Besides information theory~\cite{cover1999elements}, the JSD also plays an important role in machine learning~\cite{melville2005active,goodfellow2014generative,sutter2020multimodal}.
However, one drawback that refrains its use in practice is that the JSD between two Gaussian distributions (normal distributions) is not known in closed-form since no analytic formula is known for the differential entropy of a two-component Gaussian mixture~\cite{michalowicz2008calculation}, and thus the JSD needs to be numerically approximated in practice by various methods.

To circumvent this problem, the geometric G-JSD was defined in~\cite{GJSD-2019} as follows:

\begin{Definition}[G-JSD~\cite{GJSD-2019}]\label{def:GJSD}
The geometric Jensen--Shannon divergence (G-JSD) between two probability densities $p_1$ and $p_2$ is defined by
$$
\JS_G(p_1,p_2) \eqdef \frac{1}{2}\, \left( \KL(p_1,(p_1p_2)_G) +  \KL(p_2,(p_1p_2)_G) \right),
$$
where $(p_1p_2)_G(x)=\frac{\sqrt{p_1(x)\, p_2(x)}}{\int \sqrt{p_1(x)\, p_2(x)}\, \dmu}$ is the (normalized) geometric mixture of $p_1$ and $p_2$.
\end{Definition}

We have $\JS_G(p_1,p_2)=\KL^\JS_G(p_1,p_2)$.
Since by default the $M$- mixture  JS-symmetrization of dissimilarities $D$ are done on the right argument (i.e., $D^\JS_M$), we may also consider a dual JS-symmetrization by setting the $M$-mixtures on the left argument. We denote this left mixture JS-symmetrization by $D^{\JS^*}_M$. 
We have $D^{\JS^*}_M(p_1,p_2)=({D^*})^\JS_M(p_1,p_2)$, i.e., the left-sided JS-symmetrization of $D$ amounts to a right-sided JS-symmetrization of the dual dissimilarity $D^*(p_1,p_2)\eqdef D(p_2,p_1)$.

Thus a left-sided G-JSD divergence $\JS_G^*$ was also defined in ~\cite{GJSD-2019}:

\begin{Definition}\label{def:leftGJSD}
The left-sided geometric Jensen--Shannon divergence (G-JSD) between two probability densities $p_1$ and $p_2$ is defined by
\begin{eqnarray*}
\JS_G^*(p_1,p_2) &\eqdef& \frac{1}{2}\, \left( \KL((p_1p_2)_G,p_1) +  \KL((p_1p_2)_G,p_2) \right),\\
&=& \frac{1}{2}\, \left( \KL^*(p_1,(p_1p_2)_G) +  \KL^*(p_2,(p_1p_2)_G) \right),
\end{eqnarray*}
where $(p_1p_2)_G(x)=\frac{\sqrt{p_1(x)\, p_2(x)}}{\int \sqrt{p_1(x)\, p_2(x)}\, \dmu}$ is the (normalized) geometric mixture of $p_1$ and $p_2$.
\end{Definition}

To contrast with the numerical approximation limitation of the JSD between Gaussians, one advantage of the  geometric Jensen--Shannon divergence (G-JSD) is that it admits a closed-form expression between Gaussian distributions~\cite{GJSD-2019}. However, the G-JSD is not anymore bounded.
The G-JSD formula between Gaussian distributions has been used in several scenarii. See~\cite{deasy2020constraining,deasy2021alpha,kumari2023rds,ni2023learning,sachdeva2024uncertainty,wang2023np,serra2024computation,thiagarajan2025jensen,hanselmann2025emperror}) for a few use cases.

Let us express the G-JSD divergence using other familiar divergences.
 
\begin{Proposition}\label{prop:}
We have the following expression of the geometric Jensen--Shannon divergence:
$$
\JS_G(p_1,p_2)=\frac{1}{4}\, J(p_1,p_2) - B(p_1,p_2),
$$ 
where $J(p_1,p_2)\eqdef\int (p_1-p_2)\log\frac{p_1}{p_2}\,\dmu$ is Jeffreys' divergence~\cite{jeffreys1998theory} and 
$$
B(p_1,p_2)=-\log\int\sqrt{p_1p_2}\,\dmu=-\log Z_G(p_1,p_2),
$$ 
is the Bhattacharrya distance. 
\end{Proposition}

\begin{proof}
We have:
\begin{eqnarray*}
\JS_G(p_1,p_2) &:=& \frac{1}{2} \left(\KL(p_1,(p_1p_2)_G)+\KL(p_2,(p_1p_2)_G)\right),\\
&=& \frac{1}{2} \left(\int \left(p_1(x)\log\frac{p_1(x)\, Z_G(p_1,p_2)}{\sqrt{p_1(x)\,p_2(x)}}+ p_2(x)\log\frac{p_2(x)\, Z_G(p_1,p_2)}{\sqrt{p_1(x)\,p_2(x)}}\right)\dmu(x) \right),\\
&=& \frac{1}{2}\left(\int \left(p_1(x)+p_2(x)\right)\log Z_G(p_1,p_2)\,\dmu(x)+\frac{1}{2}\KL(p_1,p_2)+\frac{1}{2}\KL(p_2,p_1)  \right),\\
&=& \log Z_G(p_1,p_2) + \frac{1}{4} J(p_1,p_2),\\
&=& \frac{1}{4} J(p_1,p_2) - B(p_1,p_2).
\end{eqnarray*}
\end{proof}

\begin{Corollary}[G-JSD upper bound]
We have the upper bound $\JS_G(p,q)\leq \frac{1}{4} \, J(p,q)$.
\end{Corollary}

\begin{proof}
Since $B(p_1,p_2)\geq 0$ and $\JS_G(p_1,p_2)=\frac{1}{4}\, J(p_1,p_2) - B(p_1,p_2)$, we have  $\JS_G(p,q)\leq \frac{1}{4} \, J(p,q)$.
\end{proof}

\begin{Remark}
Although the KLD and JSD are separable divergences (i.e., $f$-divergences expressed as integrals of scalar divergences), the $M$-JSD divergence is in general not separable because 
it requires to normalize mixtures inside the log terms.
Notice that the Bhattacharyya distance is similarly not a separable divergence but the Bhattacharyya similarity coefficient $\BC(p_1,p_2)=\exp(-B(p_1,p_2))=\int \sqrt{p_1\, p_2}\,\dmu$ is a separable ``$f$-divergence''/$f$-coefficient for $f_\BC(u)=\sqrt{u}$ (here, a concave generator): $\BC(p_1,p_2)=I_{f_\BC}(p_1,p_2)$.
Notice that $f_\BC(1)=1$, and because of the concavity of $f_\BC$, we have $I_{f_\BC}(p_1,p_2)\leq f_\BC(1)=1$ (hence, the term $f$-coefficient to reflect the notion of similarity measure).
\end{Remark}

\subsection{Paper outline}

The paper is organized as follows:
We first give an alternative definition of the M-JSD in \S\ref{sec:ExtGJSD} (Definition~\ref{def:eMJSD}) which extends to positive measures and do not require normalization of geometric mixtures. We call this new divergence the extended M-JSD, and we compare the two types of geometric JSDs when dealing with probability measures.
In \S\ref{sec:reg}, we show that these normalized/extended $M$-JSD divergences can be interpreted as regularizations of the Jensen--Shannon divergence, and exhibit several bounds.
 We discuss Monte Carlo stochastic approximations and approximations using $\gamma$-divergences~\cite{fujisawa2008robust} in \S\ref{sec:est}.
For the case of geometric mixtures, although the G-JSD is not a $f$-divergence, we show that the extended G-JSD is a $f$-divergence (Proposition~\ref{prop:eGJSDfdiv}), and we express both the G-JSD and the extended G-JSD using both the Jeffreys divergence and the Bhattacharyya divergence or coefficient.
We report related closed-form formula for the G-JSD and extended G-JSD between two Gaussian distributions in Section~\ref{sec:gaussian}.
Finally, we summarize the main results in the concluding section \S\ref{sec:concl}.

A list of notations is provided in Appendix~\ref{sec:notations}.

\section{A novel definition G-JSD extended to positive measures}\label{sec:ExtGJSD}

\subsection{Definition and properties}
We may consider the following two modifications of the G-JSD:
\begin{itemize} 

\item First, we replace the KLD by the extended KLD 
between positive densities $q_1\in M_\mu^+$ and $q_2\in M_\mu^+$ instead of normalized densities:
\begin{equation}\label{eq:ekld}
{\KL^+}(q_1,q_2)\eqdef \int \left(q_1\log\frac{q_1}{q_2}+q_2-q_1\right) \, \dmu,
\end{equation}
 (with $\KL^+(p_1,p_2)=\KL(p_1,p_2)$) and,

\item Second, we consider unnormalized $M$-mixture densities: 
$$
{(q_1q_2)}_{\tM_\alpha}(x)\eqdef {M_\alpha(q_1(x),q_2(x))},
$$
 where we use the $\tilde{M}$ tilde notation to indicate that the $M$-mixture is not normalized, instead of normalized densities ${(q_1q_2)}_{M_\alpha}(x)$.
\end{itemize}

{Consider the KLD formula between a normalized density $p_1$ and an unnormalized density $q_2=\lambda p_1$ for some $\lambda>1$.
We have $\KL(p_1,q_2)=\int p_1\log\frac{p_1}{q_2} \,\dmu= \int p_1\log\frac{1}{\lambda}\,\dmu=-\log \lambda<0$.
That is the KLD can be negative between non-normalized  densities.
However, the extended KLD is always guaranteed to be positive for $p_1>0$ and $q_2>0$ since it can be written as a pointwise scalar Bregman divergence integral for the negative Shannon entropy generator~\cite{jones2002general}:

\begin{eqnarray*}
\KL^+(p_1,q_2) &=& \int (p_1\log\frac{p_1}{q_2} + q_2-p_1)\, \dmu,\\
&=& \int B_F(p_1(x),q_2(x))\, \dmu \geq 0,
\end{eqnarray*}
where $F(y)=y\log y-y$ is the extended Shannon negative entropy function:
$B_F(a,b)=a\log\frac{a}{b}+b-a\geq 0$ with equality if and only if $a=b$.
}

The extended KLD is an extended $f$-divergence~\cite{nishimura2008information}: 
${\KL^+}(q_1,q_2)=I_{f_{\KL^+}}^+(q_1,q_2)$ for $f_{\KL^+}(u)=-\log(u)+u-1$, 
where $I_f^+(q_1,q_2)$ denotes the  $f$ divergence extended to positive densities $q_1$ and $q_2$:
$$
I_f^+(q_1,q_2) = \int q_1\, f\left(\frac{q_2}{q_1}\right)\, \dmu.
$$

\begin{Remark}
As a side remark, it is preferable in practice to estimate the KLD between $p_1$ and $p_2$ by Monte Carlo methods using Eq.~\ref{eq:ekld} instead of Eq.~\ref{eq:kld} in order to guarantee the non-negativeness of the KLD (Gibbs' inequality).
Indeed, the sampling of $s$ samples $x_1,\ldots, x_s$, defines two unnormalized distributions $q_1(x)=\frac{1}{s}\sum_{i=1}^s p_1(x)\delta_{x_i}(x)$ and 
$q_2(x)=\frac{1}{s}\sum_{i=1}^s p_2(x)\delta_{x_i}(x)$ where 
$$
\delta_{x_i}(x)=\left\{\begin{array}{ll}1, & \mbox{if $x=x_i$}\cr 0,& \mbox{otherwise}\end{array}\right..
$$
\end{Remark}

\begin{Remark}
For an arbitrary distortion measure $D^+(q_1,q_2)$ between positive measures $q_1$ and $q_2$, we can build a corresponding projective divergence  $\tD(q_1,q_2)$ 
as follows:
$$
\tD(q_1,q_2)\eqdef D^+\left(\frac{q_1}{Z(q_1)},\frac{q_1}{Z(q_2)}\right),
$$
where $Z(q)\eqdef \int q\,\dmu$ is the normalization factor of the positive density $q$.
The divergence $\tD$ is said projective because we have for all $\lambda_1>0,\lambda_2>0$, the property that 
$\tD(\lambda_1 q_1,\lambda_2 q_2)=\tD(q_1,q_2)=D^+(p_1,p_2)$ where $p_i=\frac{q_i}{Z(q_i)}$ are the normalized densities.
The projective Kullback--Leibler divergence $\widetilde{\KL}$ is thus another projective extension of the KLD to non-normalized densities which coincide with the KLD for probability densities. But the projective KLD is different from the extended KLD of Eq.~\ref{eq:ekld}, and furthermore we have $\widetilde{\KL}(q_1,q_2)=0$ if and only if $q_1=\lambda\, q_2$ $\mu$-almost everywhere for some $\lambda>0$.
\end{Remark}

Let us now define the Jensen--Shannon symmetrization of an extended statistical divergence $D^+$ with respect to an arbitrary weighted mean $M_\alpha$ as follows:

\begin{Definition}[Extended M-JSD]\label{def:eMJSD}
A Jensen--Shannon skew symmetrization of a statistical divergence $D^+(\cdot,\cdot)$ between two positive measures $q_1$ and $q_2$ with respect to a weighted mean $M_\alpha$ is defined by
\begin{equation}\label{eq:edjs}
D^{\JS^+}_{\tM_\alpha,\beta}(q_1,q_2)
 \eqdef 
\beta \, D^+\left(q_1,{(q_1q_2)}_{\tM_\alpha}\right)
 + (1-\beta )\, D^+\left(q_1,{(q_1q_2)}_{\tM_\alpha}\right),
\end{equation}

\end{Definition}

When $\beta=\frac{1}{2}$, we write for short $D^{\JS^+}_{\tM_\alpha}(q_1,q_2)$,
and furthermore when $\alpha=\frac{1}{2}$, we simplify the notation to $D^{\JS^+}_{\tM}(q_1,q_2)$.

When $D^+=\KL^+$, we obtain the extended geometric Jensen-Shannon divergence, $\JS_{\tG}^+(q_1,q_2)=\KL^{\JS^+}_{\tG}(q_1,q_2)$:

\begin{Definition}[Extended G-JSD]\label{def:eGJSD}
The extended geometric Jensen--Shannon divergence between two positive densities $q_1$ and $q_2$ is
\begin{eqnarray}
\JS_{\tG}^+(q_1,q_2) &=& \frac{1}{2}\,\left( \KL^+(q_1,(q_1q_2)_{\tG})+ \KL^+(q_2,(q_1q_2)_{\tG}))\right),\label{eq:eGJSD}
\end{eqnarray}
\end{Definition}

The extended G-JSD between two normalized densities $p_1$ and $p_2$ is thus
\begin{eqnarray}
\JS_{\tG}^+(p_1,p_2) &=& \frac{1}{2}\left(\int \left(p_1\log\frac{p_1}{\sqrt{p_1\,p_2}} +p_2\log\frac{p_2}{\sqrt{p_1\,p_2}} \right)\,\dmu +\int \sqrt{p_1\,p_2}\, \dmu)\right) -1,\\
&=& \frac{1}{2} \left(
\int \left(p_1\log \sqrt{\frac{p_1}{p_2}}  +p_2\log \sqrt{\frac{p_2}{{p_1}}} 
\right)\,\dmu
+Z_G(p_1,p_2)
\right) -1,
\end{eqnarray}
with $Z_G(p_1,p_2)=\exp(-B(p_1,p_2))$.

Thus we get the following propositions:

\begin{Proposition}\label{prop:eGJSDid}
The extended geometric Jensen--Shannon divergence (G-JSD) can be expressed as follows:
$$
\JS_{\tG}^+(p_1,p_2) = \frac{1}{4}\, J(p_1,p_2)+\exp(-B(p_1,p_2))-1.
$$
\end{Proposition}

\begin{proof}
We have
\begin{eqnarray*}
\JS_{\tG}^+(p_1,p_2) &=& \frac{1}{2}\,\left(
\KL^+(p_1,(p_1p_2)_{\tG}) + \KL^+(p_2,(p_1p_2)_{\tG}) 
\right),\\
&=& 
\frac{1}{2}\,\left(\int \left(
p_1\log \sqrt{\frac{p_1}{p_2}} + p_2\log \sqrt{\frac{p_2}{p_1}} + 2\sqrt{p_1\, p_2}- (p_1+p_2) \right)\,\dmu 
\right),\\
&=& \int \frac{1}{4}\, (p_1-p_2)\log\frac{p_1}{p_2}\, \dmu + \int \sqrt{p_1\, p_2}\,\dmu -1,\\
&=& \frac{1}{4}\, J(p_1,p_2) + \exp(-B(p_1,p_2)) -1.
\end{eqnarray*}
\end{proof}

Thus we can express the gap between $\JS_{\tG}^+(p_1,p_2)$ and $\JS_G(p_1,p_2)$:
$$
\Delta_G(p_1,p_2)=\JS_{\tG}^+(p_1,p_2)-\JS_G(p_1,p_2)=\exp(-B(p_1,p_2))+B(p_1,p_2)-1.
$$ 

Since $Z_G(p_1,p_2)=\exp(-B(p_1,p_2))$, we have:
$$
\Delta_G(p_1,p_2)=Z_G(p_1,p_2)-\log Z_G(p_1,p_2)-1.
$$

\begin{Proposition}\label{prop:eGJSDfdiv}
The extended G-JSD is a $f$-divergence for the generator 
$$
f_{\tG}(u)=\frac{1}{4}\left(u-1\right)\log u+\sqrt{u}-1.
$$
That is, we have $\JS_{\tG}^+(p_1,p_2) = I_{f_{\tG}}(p_1,p_2)$.
\end{Proposition}

\begin{proof}
We proved that $\JS_{\tG}^+(p_1,p_2) = \frac{1}{4}\, J(p_1,p_2)+\BC(p_1,p_2)-1$.
The Jeffreys divergence is a $f$-divergence for the generator $f_J(u)=(u-1)\log u$,
and the Bhattacharyya coefficient is a $f$-coefficient for $f_\BC(u)=\sqrt{u}$ (a ``$f$-divergence'' for a concave generator).
Thus we have
$$
f_{\tG}(u)=\frac{1}{4}\left(u-1\right)\log u+\sqrt{u}-1,
$$ 
such that $\JS_{\tG}^+(p_1,p_2) = I_{f_{\tG}}(p_1,p_2)$.
We check that $f_{\tG}(u)$ is convex since $f_{\tG}''(u)=\frac{\sqrt{u}(u+1)-u}{4u^{\frac{5}{2}}}$ (and by a change of variable $t=\sqrt{u}$, the numerator $t(t^2-t+1)$ is shown positive since the discriminant of $t^2-t+1$ is negative), and we have $f_{\tG}(1)=0$.
Thus the extended G-JSD is a proper $f$-divergence.
\end{proof}

It follows that the extended G-JSD satisfies the information monotonicity of invariant divergences in information geometry~\cite{IG-2016}.

{
\begin{Remark}
More generally, let us define the extended $(\alpha,\beta)$-GJSD for $\alpha\in(0,1),\beta\in(0,1)$ as
$$
\JS_{G_\alpha,\beta}(p_1,p_2)=\int \left(
\beta\,p_1\log\frac{p_1}{p_1^\alpha\, p_2^{1-\alpha}} + (1-\beta)\,p_2\log\frac{p_2}{p_1^\alpha\, p_2^{1-\alpha}}
 + p_1^\alpha\, p_2^{1-\alpha} - (\beta p_1^\alpha +(1-\beta) p_2^{1-\alpha})
\right)\, \dmu.
$$

Then we get  the following identity:
\begin{eqnarray*}
\JS_{G_\alpha,\beta}(p_1,p_2) &=& \beta(1-\alpha) \KL(p_1,p_2)+(1-\beta)\alpha \KL(p_2,p_1)+\BC_\alpha(p_1,p_2)-1.
\end{eqnarray*}

Furthermore, divergence $\JS_{G_\alpha,\beta}$ is expressed using the $f$-divergence formula for the following  generator:
$$
f_{\alpha,\beta}(u)=-(1-\alpha)\beta\log(u)+\alpha(1-\beta)u\log(u)+u^{1-\alpha}-(\beta+(1-\beta)u).
$$

Let $\alpha=\beta$. Then we have
$$
f_{\alpha,\alpha}''(u)= \frac{\alpha(1-\alpha)}{u^{2+\alpha}}\, (u^\alpha(u+1)+u)
>0, \forall u>0,\forall \alpha\in (0,1).
$$


Hence $\JS_{G_\alpha,\beta}(p_1,p_2)=I_{f_{\alpha,\alpha}}(p_1,p_2)\geq 0$, i.e., 
the extended $(\alpha,\alpha)$-GJSD is a $f$-divergence since $f_{\alpha,\alpha}(u)$ is strictly convex
 and we have $f_{\alpha,\alpha}(1)=0$.
\end{Remark}
}

By abuse of notations, we have
$$
{\KL^+}(q_1,q_2)\eqdef \KL(q_1,q_2) + \int \left(q_2-q_1\right)\,\dmu,
$$
although $q_1$ and $q_2$ may not need to be normalized in the $\KL$ term (which can then yield a potentially negative value).
Letting $Z(q_i)\eqdef \int q_i\,\dmu$ be the total mass of positive density $q_i$, we have 
\begin{equation}
{\KL^+}(q_1,q_2) = \KL(q_1,q_2) + Z(q_2)-Z(q_1).
\end{equation}

Let $\tm_\alpha=M_\alpha(q_1,q_2)$ be the unnormalized $M$-mixture of positive densities $q_1$ and $q_2$, and set $Z_{M_\alpha}=\int \tm_\alpha\,\dmu$ be the normalization term so that we have $m_\alpha=\frac{\tm_\alpha}{Z_{M_\alpha}}$ and $\tm_\alpha=Z_{M_\alpha}\, m_\alpha$. 
When clear from context, we write $Z_\alpha$ instead of $Z_{M_\alpha}$.

We get after elementary calculus the following identity:
\begin{equation}
\JS^+_{\tilde{M}_\alpha,\beta}(q_1,q_2) = \JS_{{{M}_\alpha,\beta}}(q_1,q_2)
- (\beta Z(q_1)+(1-\beta)Z(q_2))\log Z_\alpha + Z_\alpha-(\beta Z(q_1)+(1-\beta)Z(q_2)).
\end{equation}

Therefore the difference gap $\Delta_{M_\alpha,\beta}(p_1,p_2)$ (written for short as $\Delta(p_1,p_2)$) between the normalized JSD and the unnormalized M-JSD between two normalized densities $p_1$ and $p_2$ (i.e., with $Z_1=Z(p_1)=1$ and 
$Z_2=Z(p_2)=1$) is
\begin{equation}\label{eq:gap}
\Delta(p_1,p_2) \eqdef {\JS^+}_{\tilde{M}_\alpha,\beta}(p_1,p_2)  - \JS^{{M}_\alpha,\beta}(p_1,p_2) 
 =  Z_\alpha-\log(Z_\alpha)  -1. 
\end{equation}

\begin{Proposition}[Extended/normalized M-JSD Gap]\label{prop:gap}
The following identity holds: 
$$
{\JS^+}_{\tilde{M}_\alpha,\beta}(p_1,p_2)=\JS_{{M}_\alpha,\beta}(p_1,p_2)+Z_\alpha-\log(Z_\alpha)  -1.
$$
\end{Proposition}

Thus ${\JS^+}_{\tilde{M}_\alpha,\beta}(p_1,p_2)\geq \JS_{{M}_\alpha,\beta}(p_1,p_2)$ when $\Delta(p_1,p_2)\geq 0$
and ${\JS^+}_{\tilde{M}_\alpha,\beta}(p_1,p_2)\leq \JS_{{M}_\alpha,\beta}(p_1,p_2)$ when $\Delta(p_1,p_2)\leq 0$.
 
When we consider the weighted arithmetic mean $A_\alpha$, we always have $Z_\alpha=1$ for $\alpha\in(0,1)$, and thus the two definitions (Definition~\ref{def:MJSD} and Definition~\ref{def:eMJSD}) of the $A$-JSD coincide (i.e., $Z_\alpha^A-\log(Z_\alpha^A)  -1=0$): 
$$
\JS_A(p_1,p_2)={\JS}_{\tilde{A}}(p_1,p_2).
$$
However, when the weighted mean $M_\alpha$ differs from the weighted arithmetic mean (i.e., $M_\alpha\not=A_\alpha$), the two definitions of the M-JSD 
$\JS_M$ and extended M-JSD ${\JS}_{\tilde{M}}$ differ by the gap expressed in Eq.~\ref{eq:gap}.

\begin{Remark}\label{rk:logbase}
When information is measured in bits, logarithms are taken to base $2$ and when information is measured in nats, base $e$ is considered.
Thus we shall generally consider  the gap $\Delta_b=Z_\alpha-\log_b(Z_\alpha)  -1$ where $b$ denotes the base of the logarithm.
When $b=e$, we have $\Delta_e\geq 0$ for all $Z_\alpha>0$.
When $b=2$, we have $\Delta_2=Z_\alpha-\log_2(Z_\alpha)  -1\geq 0$ when $0<Z_\alpha\leq 1$ or $Z_\alpha\geq 2$.
But since $Z_\alpha\leq 2$ (see Eq.~\ref{eq:zub}), the condition simplifies to $\Delta_2\geq 0$ if and only if $Z_\alpha\leq 1$. 
\end{Remark}

\begin{Remark}\label{rk:gjsdnotmetric}
Although $\sqrt{\JS}$ is a metric distance~\cite{fuglede2004jensen}, $\sqrt{\JS_G}$ is not a metric distance as the triangle inequality is not satisfied.
It suffices to report a counterexample of the triangle inequality for a triple of points $p_1$, $p_2$, and $p_3$:
Consider $p_1=(0.55,0.45)$, $p_2=(0.002,0.998)$, and $p_3=(0.045,0.955)$.
Then we have
$\sqrt{\JS_G}(p_1,p_2)\approx 1.0263227\ldots$,  $\sqrt{\JS_G}(p_1,p_3)\approx 0.63852342\ldots$, and 
$\sqrt{\JS_G}(p_3,p_2)\approx 0.19794622\ldots$.
The triangle inequality fails with an error of 
$$
\sqrt{\JS_G}(p_1,p_2)-(\sqrt{\JS_G}(p_1,p_3)+\sqrt{\JS_G}(p_3,p_2))\approx 0.1898531\dots.
$$

Similarly, the  triangle inequality also fails for the extended G-JSD: We have  
$\sqrt{\JS^+_G(p_1,p_2)}\approx 1.0788275\dots$, $\sqrt{\JS^+_G(p_1,p_3)} \approx  0.6691922\ldots$, and 
$\sqrt{\JS^+_G(p_3,p_2)}\approx 0.1984633\ldots$
with a triangle inequality defect value of 
$$
\sqrt{\JS^+_G(p_1,p_2)}-(\sqrt{\JS^+_G(p_1,p_3)}+\sqrt{\JS^+_G(p_3,p_2)})\approx 0.2111719\ldots.
$$
\end{Remark}

\subsection{Power JSDs and (extended) min-JSD and max-JSD}\label{sec:limits}

Let $P_{\gamma,\alpha}(a,b)\eqdef(\alpha a^\gamma+(1-\alpha)b^\gamma)^{\frac{1}{\gamma}}$ be the $\gamma$-power mean for $\gamma\not=0$ (with $A_\alpha=P_{1,\alpha}$).
Further define $P_{0,\alpha}(a,b)=G_\alpha(a,b)$ so that $P_{\gamma,\alpha}$ defines the weighted power means for  $\gamma\in\bbR$ and $\alpha\in (0,1)$ in the reminder.
Since $P_{\gamma,\alpha}(a,b)\leq P_{\gamma',\alpha}(a,b)$ when $\gamma'\geq\gamma$ for any $a,b>0$, we have that
\begin{equation}
Z_\alpha^{P_\gamma}(p_1,p_2)=\int P_{\gamma,\alpha}(p_1(x),p_2(x))\,\dmu \leq Z_\alpha^{P_{\gamma'}}(p_1,p_2)=\int P_{\gamma',\alpha}(p_1(x),p_2(x))\,\dmu.
\end{equation}

Let $P_\gamma(a,b)=P_{\gamma,\frac{1}{2}}(a,b)$.
We have $\lim_{\gamma\rightarrow -\infty} P_\gamma(a,b)=\min(a,b)$ and $\lim_{\gamma\rightarrow +\infty} P_\gamma(a,b)=\max(a,b)$.
Thus we can define both (extented) $\min$-JSD and (extented) $\max$-JSD.
Using the fact that $\min(a,b)=\frac{a+b}{2}-\frac{1}{2}|a-b|$ and  $\max(a,b)=\frac{a+b}{2}+\frac{1}{2}|a-b|$, we obtain the extremal mixture normalization terms as follows:
\begin{eqnarray}
Z_{\min}(p_1,p_2) &=& \int \min(p_1,p_2) \,\dmu = 1 - \TV(p_1,p_2),\\
Z_{\max}(p_1,p_2) &=& \int \max(p_1,p_2) \,\dmu = 1 + \TV(p_1,p_2), 
\end{eqnarray}
where $\TV(p_1,p_2)=\frac{1}{2}\int |p_1-p_2| \,\dmu$ is the total variation distance.

\begin{Proposition}[$\max$-JSD]
The following upper bound holds for $\max$-JSD:
\begin{equation}
0\leq {\JS^+}_{\widetilde{\max}}(p_1,p_2)\leq \TV(p_1,p_2).
\end{equation}
\
Furthermore, the following identity relates the two types of $\max$-JSDs: 
\begin{equation}
{\JS^+}_{\widetilde{\max}}(p_1,p_2)={\JS}_{\widetilde{\max}}(p_1,p_2)+ \TV(p_1,p_2) -\log \left(1+\TV(p_1,p_2)\right).
\end{equation}
\end{Proposition}

\begin{proof}
We have
$$
{\JS^+}_{\widetilde{\max}}(p_1,p_2) \eqdef 
\frac{1}{2}\, \int \left(p_1\log\frac{p_1}{\max(p_1,p_2)} 
  + p_2\log\frac{p_2}{\max(p_1,p_2)} + 2\max(p_1,p_2)-(p_1+p_2)\right)\, \dmu. 
$$

Since both $\log\frac{p_1}{\max(p_1,p_2)}\leq 0$ and $\log\frac{p_2}{\max(p_1,p_2)}\leq 0$, and $\max(a,b)=\frac{a+b}{2}+\frac{1}{2}|b-a|$, we have
$$
{\JS^+}_{\widetilde{\max}}(p_1,p_2) \leq  \int \left( \frac{p_1+p_2}{2}+\frac{1}{2}|p_2-p_1| - \frac{p_1+p_2}{2}\right)\, \dmu.
$$
That is, ${\JS^+}_{\widetilde{\max}}(p_1,p_2) \leq   \TV(p_1,p_2)$.

We characterize the gap as follows: 
\begin{eqnarray*}
\Delta_{\max}(p_1,p_2) &=& Z_{\max}(p_1,p_2)-\log Z_{\max}(p_1,p_2)-1,\\
 &=&   \TV(p_1,p_2) -\log (1+\TV(p_1,p_2)) \geq 0,
\end{eqnarray*}
since $0\leq \TV\leq 1$.
Thus ${\JS^+}_{\widetilde{\max}}(p_1,p_2) \geq {\JS}_{{\max}}(p_1,p_2)$.
\end{proof}

\begin{Proposition}[min-JSD]\label{prop:minJSD}
We have the following lower bound on the extended min-JSD:
$$
{\JS^+}_{\widetilde{\min}}(p_1,p_2)\geq  \frac{1}{4} \, J(p_1,p_2)- \TV(p_1,p_2),
$$
where $J(p_1,p_2)\eqdef \KL(p_1,p_2)+\KL(p_2,p_1) =\int (p_1-p_2)\log\frac{p_1}{p_2}\,\dmu$ is Jeffreys' divergence~\cite{jeffreys1998theory}
and
$$
{\JS^+}_{\widetilde{\min}}(p_1,p_2) = {\JS}_{{\min}}(p_1,p_2) -\TV(p_1,p_2)+\log (1-\TV(p_1,p_2)).
$$
\end{Proposition}

\begin{proof}
We have  $Z_{\min}(p_1,p_2)=\int \min\{p_1,p_2\}\,\dmu = 1-\TV(p_1,p_2)\leq 1$ and
\begin{eqnarray*}
\Delta_{\min}(p_1,p_2) &=& Z_{\min}(p_1,p_2)-\log Z_{\min}(p_1,p_2)-1,\\
 &=&   -\TV(p_1,p_2) -\log (1-\TV(p_1,p_2)) \geq 0,
\end{eqnarray*}
since $-x-\log(1-x)\geq 0$ for $x\leq 1$.
Note that the gap can be arbitrarily large when $\TV(p_1,p_2)\rightarrow 1^-$.

Thus we have ${\JS^+}_{\widetilde{\min}}(p_1,p_2)\geq {\JS}_{{\min}}(p_1,p_2)$.

To get the lower bound, we use the fact that $\min(p_1,p_2)\leq \sqrt{p_1p_2}$.
Indeed, we have
\begin{eqnarray*}
{\JS^+}_{\widetilde{\min}}(p_1,p_2) &=& \frac{1}{2} \left(\int (p_1\log \frac{p_1}{\min(p_1,p_2)}+ p_2\log \frac{p_2}{\min(p_1,p_2)} 
 +2 \min(p_1,p_2) - (p_1+p_2) \right)\, \dmu,\\
&\geq & \frac{1}{2} \int \left(  \frac{1}{2} p_1\log \frac{p_1}{p_2}+ \frac{1}{2} p_2\log \frac{p_2}{p_1} 
 +2 \min(p_1,p_2) - (p_1+p_2) \right)\, \dmu,\\
&= & \frac{1}{4} \, J(p_1,p_2) - \TV(p_1,p_2).
\end{eqnarray*}
\end{proof}

\begin{Remark}
Let us report the total variation distance between two univariate Gaussian distributions $p_{\mu_1,\sigma_1}$ and $p_{\mu_2,\sigma_2}$
 in closed-form using the error function~\cite{nielsen2014generalized}: $\erf(x)=\frac{1}{\sqrt{\pi}} \int_{-x}^x e^{-t^2} \dt$.
\begin{itemize}
\item When $\sigma_1=\sigma_2=\sigma$, we have
\begin{equation}
\TV(p_1,p_2) = \frac{1}{2} \left| \Phi(x^*;\mu_2,\sigma)-  \Phi(x^*;\mu_1,\sigma)\right|,
\end{equation}
where $\Phi(x;\mu,\sigma)=\frac{1}{2}(1+\erf(\frac{x-\mu}{\sigma\sqrt{2}}))$ is the cumulative distribution, and
\begin{equation}
x^*= \frac{\mu_1^2-\mu_2^2}{2(\mu_1-\mu_2)}.
\end{equation}

\item When $\sigma_1\not=\sigma_2$, we let $x_1=\frac{-b-\sqrt{\Delta}}{2a}$ and $x_2=\frac{-b+\sqrt{\Delta}}{2a}$ where 
$\Delta=b^2-4ac\geq 0$ and

\begin{eqnarray}
a&=& \frac{1}{\sigma_1^2}-\frac{1}{\sigma_2^2},\\
b&=&  2\left(\frac{\mu_2}{\sigma_2}-\frac{\mu_1}{\sigma_1}\right),\\
c&=&  \left(\frac{\mu_1}{\sigma_1}\right)^2 - \left(\frac{\mu_2}{\sigma_2}\right)^2 -2\log \frac{\sigma_2}{\sigma_1}.
\end{eqnarray}

The total variation is given by
\begin{eqnarray}\label{eq:unitv}
\lefteqn{\TV( p_1,  p_2) =}\nonumber \\
&& \frac{1}{2} \left( 
\left|
\erf\left(\frac{x_1-\mu_1}{\sigma_1\sqrt{2}}\right)-\erf\left(\frac{x_1-\mu_2}{\sigma_2\sqrt{2}}\right)
\right|
+
\left|
\erf\left(\frac{x_2-\mu_1}{\sigma_1\sqrt{2}}\right)-\erf\left(\frac{x_2-\mu_2}{\sigma_2\sqrt{2}}\right)
\right|
\right)
\end{eqnarray}
\end{itemize}

\end{Remark}

Next, we shall consider the important case of $p_1$ and $p_2$ belonging to the family of multivariate normal distributions, commonly called Gaussian distributions.

\section{Geometric JSDs between Gaussian distributions}\label{sec:gaussian}

\subsection{Exponential families}

The formula for the G-JSD between two Gaussian distributions was reported in~\cite{GJSD-2019} using the more general framework of exponential families.
An exponential family~\cite{barndorff2014information} is a family of probability measures $\{P_\lambda\}$ with Radon-Nikodym densities $p_\lambda$ with respect to $\mu$ expressed canonically as
\begin{eqnarray*}
p_\lambda(x) &\eqdef&\exp\left(\inner{\theta(\lambda)}{t(x)}-F(\theta)+k(x)\right),\\
&=& \frac{1}{Z(\theta)}\,\exp\left(\inner{\theta(\lambda)}{t(x)}+k(x)\right),
\end{eqnarray*}
where $\theta(\lambda)$ is the natural parameter, $t(x)$ the sufficient statistic, $k(x)$ an auxiliary carrier term with respect to $\mu$, and $F(\theta)$ the cumulant function. The partition function $Z(\theta)$ is the normalizer denominator: $Z(\theta)=\exp(F(\theta))$.
The cumulant function  $F(\theta)=\log Z(\theta)$ is strictly convex and analytic~\cite{barndorff2014information}, and the partition function $Z(\theta)=\exp(F(\theta))$ is strictly log-convex (and hence also strictly convex).

We consider the exponential family of multivariate Gaussian distributions 
$$
\calN=\{N(\mu,\Sigma) \st \mu\in\bbR^d, \Sigma\in\PD(d)\},
$$ 
where $\PD(d)$ denotes the set of symmetric positive-definite matrices of size $d\times d$.
Let $\lambda \eqdef(\lambda_v,\lambda_M)=(\mu,\Sigma)$ denote the compound (vector,matrix) parameter of a Gaussian.
The $d$-variate Gaussian density is given by
\begin{eqnarray}\label{eq:mvnl}
p_\lambda(x;\lambda) &\eqdef&  \frac{1}{(2\pi)^{\frac{d}{2}}\sqrt{|\lambda_M|}}  \exp\left(-\frac{1}{2} (x-\lambda_v)^\top \lambda_M^{-1} (x-\lambda_v)\right),
\end{eqnarray} 
where $|\cdot|$ denotes the matrix determinant.
The natural parameters $\theta$ are expressed using both a {vector parameter}  $\theta_v$ and a  {matrix parameter}  $\theta_M$
 in a compound parameter $\theta=(\theta_v,\theta_M)$.
By defining the following {compound inner product} on a compound (vector,matrix) parameter
\begin{equation}
\inner{\theta}{\theta'}\eqdef \theta_v^\top \theta_v'+ \tr\left({\theta_M'}^\top\theta_M\right),
\end{equation}
where $\tr(\cdot)$ denotes the matrix trace, we rewrite the Gaussian density of Eq.~\ref{eq:mvnl} in the canonical form of an exponential family:
\begin{eqnarray}
 p_\theta(x;\theta) &\eqdef& \exp\left(\inner{t(x)}{\theta}-F_\theta(\theta)\right) = p_\lambda(x),
\end{eqnarray} 
where $\theta=\theta(\lambda)$ with
\begin{equation}
\theta=(\theta_v,\theta_M)=\left(\Sigma^{-1}\mu,-\frac{1}{2}\Sigma^{-1}\right)=\theta(\lambda)=\left(\lambda_M^{-1}\lambda_v,-\frac{1}{2}\lambda_M^{-1}\right),
\end{equation}
 is the {compound vector-matrix natural parameter} and 
 \begin{equation}
t(x)=(x,-xx^\top),
\end{equation}
 is the {compound vector-matrix sufficient statistic}. There is no auxiliary carrier term (i.e., $k(x)=0$).
The  function $F_\theta$ is given by:
\begin{equation}
F_\theta(\theta) \eqdef \frac{1}{2}\left( d\log\pi -\log |\theta_M|+\frac{1}{2} \theta_v^\top \theta_M^{-1} \theta_v \right),
\end{equation}

\begin{Remark}
Beware that when the cumulant function is expressed using the ordinary parameter $\lambda=(\mu,\Sigma)$, the cumulant function $F_\theta(\theta(\lambda))$ is not anymore convex:
\begin{eqnarray}
F_\lambda(\lambda) &=& \frac{1}{2}\left(\lambda_v^\top \lambda_M^{-1}\lambda_v+\log |\lambda_M| + d\log2\pi \right),\\
&=& \frac{1}{2}\left(\mu^\top \Sigma^{-1}\mu +\log |\Sigma| + d\log2\pi \right).
\end{eqnarray}
\end{Remark}

We convert between the ordinary parameterization $\lambda=(\mu,\Sigma)$ and the natural parameterization $\theta$ using these formula:
 
\begin{eqnarray*}
\theta=(\theta_v,\theta_M)=\left\{
\begin{array}{ll}
\theta_v(\lambda)=\lambda_M^{-1}\lambda_v=\Sigma^{-1}\mu\\
\theta_M(\lambda)=\frac{1}{2}\lambda_M^{-1}=\frac{1}{2}\Sigma^{-1}
\end{array}
\right.  &\Leftrightarrow &
\lambda=(\lambda_v,\lambda_M)=\left\{
\begin{array}{ll}
\lambda_v(\theta) = \frac{1}{2}\theta_M^{-1}\theta_v=\mu\\
\lambda_M(\theta) = \frac{1}{2}\theta_M^{-1}=\Sigma
\end{array}
\right.\\
\end{eqnarray*}

The geometric mixture $p_{\theta_1}^\alpha p_{\theta_2}^{1-\alpha}$ of two densities of an exponential family is a density $p_{\alpha\theta_1+(1-\alpha)\theta_2}$ of the exponential family with partition function $Z_\alpha(\theta_1,\theta_2)=\exp(-J_{F,\alpha}(\theta_1,\theta_2))$ where $J_{F,\alpha}(\theta_1,\theta_2)$ denotes the skew Jensen divergence~\cite{kailath1967divergence,nielsen2011burbea}:
$$
J_{F,\alpha}(\theta_1,\theta_2) \eqdef \alpha F(\theta_1)+(1-\alpha)F(\theta_2)-F(\alpha\theta_1+(1-\alpha)\theta_2).
$$

Therefore the difference gap of Eq.~\ref{eq:gap} between the G-JSD and the extended G-JSD between exponential family densities is given by:
\begin{eqnarray}
\Delta(\theta_1,\theta_2) &=& \exp({-J_{F,\alpha}(\theta_1,\theta_2)})+J_{F,\alpha}(\theta_1,\theta_2)-1,\\
&=& Z_\alpha(\theta_1,\theta_2)-\log Z_\alpha(\theta_1,\theta_2) -1,\\
&=& Z_\alpha(\theta_1,\theta_2)-F(\alpha\theta_1+(1-\alpha)\theta_2)-1.
\end{eqnarray}

Since $Z_\alpha=\exp(-J_{F,\alpha}(\theta_1,\theta_2))\leq 1$, the gap $\Delta$ is negative, and we have
$$
{\JS^+}_{\tilde{G}_\alpha,\beta}(p_{\mu_1,\Sigma_1},p_{\mu_2,\Sigma_2})  \leq  \JS_{{G}_\alpha,\beta}(p_{\mu_1,\Sigma_1},p_{\mu_2,\Sigma_2}).
$$

\begin{Corollary}\label{cor:gjsdef}
When $p_1=p_{\theta_1}$ and $p_2=p_{\theta_2}$ belongs to a same exponential family with cumulant function $F(\theta)$, we have

\begin{equation}
\JS_G(p_{\theta_1},p_{\theta_2}) =
 {\frac{1}{4}(\theta_2-\theta_1)^\top(\nabla F(\theta_2)-\nabla F(\theta_1))}-{\left(
\frac{F(\theta_1)+F(\theta_2)}{2}-F\left(\frac{\theta_1+\theta_2}{2}\right)
\right)},
\end{equation}
since $J(p_{\theta_1},p_{\theta_2})=\inner{\theta_2-\theta_1}{\nabla F(\theta_2)-\nabla F(\theta_1)}$ amounts to a symmetrized Bregman divergence.
\end{Corollary}

\begin{proof}
We have $J(p_{\theta_1},p_{\theta_2})=(\theta_2-\theta_1)^\top(\nabla F(\theta_2)-\nabla F(\theta_1))$ and 
$J(p_{\theta_1},p_{\theta_2})=J_F(\theta_1,\theta_2)$.
\end{proof}

The extended geometric Jensen--Shannon divergence and geometric Jensen--Shannon divergence between two densities of an exponential family is given by
\begin{eqnarray}
\JS_G(p_{\theta_1},p_{\theta_2}) &=& {\frac{1}{4}(\theta_2-\theta_1)^\top(\nabla F(\theta_2)-\nabla F(\theta_1))}-{\left(
\frac{F(\theta_1)+F(\theta_2)}{2}-F\left(\frac{\theta_1+\theta_2}{2}\right)
\right)}, \label{eq:GJSD-EF}\\
\JS_{\tilde{G}}(p_{\theta_1},p_{\theta_2}) &=& \frac{1}{4}\inner{\theta_2-\theta_1}{\nabla F(\theta_2)-\nabla F(\theta_1)}-\exp(-J_F(\theta_1,\theta_2))-1,\label{eq:ExtGJSD-EF} \\
{\JS^*}_G(p_{\theta_1},p_{\theta_2}) &=& J_F(\theta_1,\theta_2)
\end{eqnarray}

\begin{Remark}
Given two densities $p_1$ and $p_2$, the family $\calG$ of geometric mixtures $\{(p_1p_2)_{G_\alpha}\propto p_1^\alpha\, p_2^{1-\alpha} \st \alpha\in(0,1)\}$ forms a 1D exponential family that has been termed likelihood ratio exponential family~\cite{grunwald2007minimum} (LREF).  
The cumulant function of this LREF is $F(\alpha)=-B_\alpha(p_1,p_2)$.
Hence, $\calG$ has also been  called a Bhattacharyya arc or Hellinger arc in the literature~\cite{cena2007exponential}.
However, notice that $\KL(p_i:(p_1p_2)_{G_\alpha})$ does not amount necessarily to a Bregman divergence because neither $p_1$ nor $p_2$ belongs to $\calG$.
\end{Remark}

\subsection{Closed-form formula for Gaussian distributions}

Let us report the corresponding closed-form formula for $d$-variate Gaussian distributions.

When $\alpha=\frac{1}{2}$, we proved that $\JS_G(p_1,p_2) = \frac{1}{4} \, J(p_1,p_2)-B(p_1,p_2)$ and 
  $\JS_{\tG}^+(p_1,p_2) = \frac{1}{4} \, J(p_1,p_2) +\exp(-B(p_1,p_2))-1$ where $\BC(p_1,p_2)=\exp(-B(p_1,p_2))$.
	Thus for the case of balanced geometric mixtures, we need to report the closed-form for the Jeffreys and Bhattacharyya distances:
	\begin{eqnarray*}
J(p_{\mu_1,\Sigma_1},p_{\mu_2,\Sigma_2}) &=& 
\frac{1}{2}\,\left(
\tr\left(\Sigma_1\Sigma_2^{-1}+\Sigma_2\Sigma_1^{-1}\right)
+ (\mu_1-\mu_2)^\top(\Sigma_1^{-1}+\Sigma_2^{-1})(\mu_1-\mu_2)-2d
\right)\
,\\	
B(p_{\mu_1,\Sigma_1},p_{\mu_2,\Sigma_2}) &=& 
\frac{1}{8}\,(\mu_1 - \mu_2)^\top \bar\Sigma^{-1} (\mu_1 - \mu_2)
+ \frac{1}{2}\log\left(
\frac{\det \bar\Sigma}{\sqrt{\det \Sigma_1 \, \det \Sigma_2}}
\right),
	\end{eqnarray*}
	where $\bar\Sigma = \frac{1}{2}\, \left(\Sigma_1 + \Sigma_2\right)$.

Otherwise, for arbitrary weighted geometric mixture $G_\alpha$,
define $(\theta_1\theta_2)_\alpha=\alpha\theta_1+(1-\alpha)\theta_2$, the weighted linear interpolation of the natural parameters $\theta_1$ and $\theta_2$.

\begin{Corollary}\label{cor:GJSDGaussian}
The skew $G$-Jensen--Shannon divergence $\JS_\alpha^G$ and the dual skew $G$-Jensen--Shannon divergence ${\JS^*}_\alpha^G$ between two $d$-variate Gaussian distributions $N(\mu_1,\Sigma_1)$ and $N(\mu_2,\Sigma_2)$ is
\begin{eqnarray*}
{\JS}_{G_\alpha}(p_{(\mu_1,\Sigma_1)},p_{(\mu_2,\Sigma_2)}) 
&=& \alpha\,\KL(p_{(\mu_1,\Sigma_1)},p_{(\mu_\alpha,\Sigma_\alpha)})+ (1-\alpha)\, \KL(p_{(\mu_2,\Sigma_2)},p_{(\mu_\alpha,\Sigma_\alpha)}),\\
&=& \alpha\, B_F((\theta_1\theta_2)_\alpha,\theta_1)+
(1-\alpha)\, B_F((\theta_1\theta_2)_\alpha,\theta_2),\\
&=& 
\frac{1}{2} \left(\tr\left(\Sigma_{\alpha}^{-1}(\alpha\Sigma_{1}+(1-\alpha)\Sigma_{2})\right) + \log\left(\frac{|\Sigma_{\alpha}|}{|\Sigma_{1}|^{\alpha}|\Sigma_{2}|^{1-\alpha}}\right)\right.\\
&=& \left. + \alpha(\mu_{\alpha}-\mu_{1})^{\top}\Sigma_{\alpha}^{-1}(\mu_{\alpha}-\mu_{1}) + (1-\alpha)(\mu_{\alpha}-\mu_{2})^{\top}\Sigma_{\alpha}^{-1}(\mu_{\alpha}-\mu_{2}) - d\right)
\\
{\JS}_{G_\alpha}^*(p_{(\mu_1,\Sigma_1)},p_{(\mu_2,\Sigma_2)}) &=& (1-\alpha)\, \KL(p_{(\mu_\alpha,\Sigma_\alpha)},p_{(\mu_1,\Sigma_1)})
+\alpha \KL(p_{(\mu_\alpha,\Sigma_\alpha)},p_{(\mu_2,\Sigma_2)}),\\
&=& \alpha\, B_F(\theta_1,(\theta_1\theta_2)_\alpha)+
(1-\alpha)\,  B_F(\theta_2,(\theta_1\theta_2)_\alpha),\\
&=& J_{F,\alpha}(\theta_1,\theta_2)=:{B_\alpha}(p_{(\mu_1,\Sigma_1)},p_{(\mu_2,\Sigma_2)}),\\
&=& \frac{1}{2}\left(\alpha\mu_1^\top\Sigma_1^{-1}\mu_1+(1-\alpha)\mu_2^\top\Sigma_2^{-1}\mu_2-\mu_\alpha^\top\Sigma_\alpha^{-1}\mu_\alpha
+\log \frac{|\Sigma_1|^{\alpha}|\Sigma_2|^{1-\alpha}}{|\Sigma_\alpha|}\right),\nonumber\\
F(\mu,\Sigma) &=& \frac{1}{2}\left(\mu^\top \Sigma^{-1}\mu +\log |\Sigma| + d\log2\pi \right),\\
F(\theta_v,\theta_M)&=& \frac{1}{2}\left( d\log\pi -\log |\theta_M|+\frac{1}{2} \theta_v^\top \theta_M^{-1} \theta_v \right),\\
\Delta(\theta_1,\theta_2) &=& \exp({-J_{F,\alpha}(\theta_1,\theta_2)})+J_{F,\alpha}(\theta_1,\theta_2)-1,
\end{eqnarray*}

where $\Sigma_\alpha$ is the matrix harmonic barycenter:

\begin{equation}
\Sigma_\alpha=  \left(\alpha\Sigma_1^{-1}+(1-\alpha) \Sigma_2^{-1}\right)^{-1},
\end{equation}
 and 
\begin{equation}
\mu_\alpha=\Sigma_\alpha \left(\alpha\Sigma_1^{-1}\mu_1+(1-\alpha) \Sigma_2^{-1}\mu_2\right).
\end{equation}
\end{Corollary}

\section{Extended and normalized G-JSDs as regularizations of the ordinary JSD}\label{sec:reg}

The $M$-Jensen--Shannon divergence $\JS_M(p,q)$ can be interpreted as a regularization of the ordinary JSD:

\begin{Proposition}[JSD regularization]\label{prop:regMJSD}
For any arbitrary mean $M$, the following identity holds:
\begin{equation}\label{eq:jsdreg}
\JS_M(p_1,p_2) = \JS(p_1,p_2) + \KL\left(\frac{p_1+p_2}{2},(p_1p_2)_M\right).
\end{equation}
\end{Proposition}
Notice that $(p_1p_2)_A=\frac{p_1+p_2}{2}$.

\begin{proof}
We have
\begin{eqnarray*}
\JS_M(p_1,p_2) &\eqdef &\frac{1}{2}\left(\KL(p_1,(p_1p_2)_M)+\KL(p_2,(p_1p_2)_M) \right),\\
&=& \frac{1}{2}\int\left( p_1\log\frac{p_1\,(p_1p_2)_A}{(p_1p_2)_M\, (p_1p_2)^A } + p_2\log\frac{p_2\, (p_1p_2)_A}{(p_1p_2)_M\, (p_1p_2)_A} \right)\, \dmu,\\
&=& \frac{1}{2} \int \left( 
 p_1\log\frac{p_1}{(p_1p_2)_A} + p_1\log\frac{(p_1p_2)_A}{(p_1p_2)_M}
 + p_2\log\frac{p_2}{(p_1p_2)_A}+p_2\log\frac{(p_1p_2)_A}{(p_1p_2)_M} 
\right) \, \dmu,\\
&=& \frac{1}{2}\int\left(p_1\log\frac{p_1}{(p_1p_2)_A}+p_2\log\frac{p_2}{(p_1p_2)_A}\right) \,\dmu 
+\int \frac{1}{2}\, (p_1+p_2)\log \frac{(p_1p_2)_A}{(p_1p_2)_M} \,\dmu,\\
&=& \JS(p_1,p_2) + \int (p_1p_2)_A\log \frac{(p_1p_2)_A}{(p_1p_2)_M} \,\dmu,\\
&=& \JS(p_1,p_2) + \KL((p_1p_2)_A,(p_1p_2)_M).
\end{eqnarray*}
\end{proof}

\begin{Remark}\label{rk:klmn}
One way to symmetrize the KLD is to consider two distinct symmetric means $M_1(a,b)=M_1(b,a)$ and $M_2(a,b)=M_2(b,a)$ and define
$$
\KL_{M_1,M_2}(p_1,p_2) = \KL((p_1p_2)_{M_1},(p_1p_2)_{M_2}) = \KL_{M_1,M_2}(p_2,p_1).
$$

We notice that $\sqrt{\KL^{A,G}}$ is not a metric distance by reporting a triple of points $(p_1,p_2,p_3)$ that fails the triangle inequality.
Consider $p_1=(0.55,0.45)$, $p_2=(0.002,0.998)$, and $p_3=(0.045,0.955)$.
We have $\sqrt{\KL_{M_1,M_2}(p_1,p_2)}=0.5374165\ldots$,
$\sqrt{\KL_{M_1,M_2}(p_1,p_3)}=0.1759400\ldots$, and
$\sqrt{\KL_{M_1,M_2}(p_3,p_2)}=0.08485931\ldots$.
The triangle inequality defect is
$$
\sqrt{\KL_{M_1,M_2}(p_1,p_2)}-(\sqrt{\KL_{M_1,M_2}(p_1,p_3)}+\sqrt{\KL_{M_1,M_2}(p_3,p_2)})=0.2766171\ldots.
$$

We can also similarly symmetrize the extended KLD as follows:
$$
\KL_{\tM_1,\tM_2}^+(q_1,q_2) = \KL^+((q_1q_2)_{\tM_1},(q_1q_2)_{\tM_2}) = \KL_{\tM_1,\tM_2}(q_2,q_1).
$$

In particular, when $M_1=A$ and $M_2=G$, we get the $\KL_{A,M}$ divergence:
$$
\KL_{A,M}(p_1,p_2)=\frac{p_1+p_2}{2} \, \log \frac{p_1+p_2}{2\sqrt{p_1p_2}}+\log Z_G(p_1,p_2),
$$
which is related to Taneja $T$-divergence~\cite{taneja1995new}:
\begin{equation}\label{eq:td}
T(p_1,p_2)=\int \frac{p_1+p_2}{2}\log \frac{p_1+p_2}{2\sqrt{p_1p_2}}.
\end{equation}
The $T$-divergence is a $f$-divergence~\cite{AliSilvey-1966,Csiszar-1967} obtained for the generator $f_T(u)=\frac{1+u}{2}\log\frac{1+u}{2\sqrt{u}}$.
\end{Remark}

\begin{Corollary}[JSD lower bound on  $M$-JSD]\label{cor:MJSDLB}
We have $\JS_M(p,q)\geq \JS(p,q)$.
\end{Corollary}

\begin{proof}
Since $\JS_M(p,q) = \JS(p,q) + \KL\left(\frac{p+q}{2},(pq)_M\right)$ and $\KL\geq 0$ by Gibbs' inequality, we have
$\JS_M(p,q)\geq \JS(p,q)$.
\end{proof}

Since the extended $M$-JSD  is
 $\JS^+_{\tilde{M}_\alpha,\beta}(p_1,p_2) = \JS_{{M}_\alpha,\beta}(p_1,p_2) +Z_\alpha-\log(Z_\alpha)  -1$,
the extended $M$-JSD $\JS^+_{\tilde{M}_\alpha,\beta}$ can also be interpreted as another regularization of the Jensen--Shannon divergence when dealt with probability densities:

\begin{equation}
\JS^+_{\tilde{M}_\alpha,\beta}(p_1,p_2) = \JS(p_1,p_2) + \KL\left(\frac{p_1+p_2}{2},(p_1p_2)_M\right)+ Z_{M_\alpha}(p_1,p_2)-\log(Z_{M_\alpha}(p_1,p_2))  -1.
\end{equation}

It is well-known that the JSD can be rewritten as a diversity index~\cite{JSD-1991} using the concave entropy:
\begin{equation}\label{eq:jsdh}
\JS(p_1,p_2)= H\left(\frac{p_1+p_2}{2}\right)-\frac{H(p_1)+H(p_2)}{2}.
\end{equation}

We generalize this decomposition as a difference of a cross-entropy term minus entropies as follows:

\begin{Proposition}[M-JSD cross-entropy decomposition]
We have
$$
\JS_M(p_1,p_2)=H^\times((p_1p_2)_A,(p_1p_2)_M)-\frac{H(p_1)+H(p_2)}{2}.
$$
\end{Proposition}

\begin{proof}
We have from Proposition~\ref{prop:regMJSD}:
$$
\JS_M(p_1,p_2) = \JS(p_1,p_2) + \KL\left(\frac{p_1+p_2}{2},(p_1p_2)_M\right).
$$
Since $\KL(p_1,p_2)=H^\times(p_1,p_2)-H(p_1)$ where $H^\times(p_1,p_2)=-\int p_1\log p_2\,\dmu$ is the cross-entropy and $H(p)=-\int p\log p\,\dmu=H^\times(p,p)$ is the entropy.
Pluggin Eq.~\ref{eq:jsdh} in Eq.~\ref{eq:jsdreg}, we get
\begin{eqnarray*}
\JS_M(p_1,p_2)  &=& H\left(\frac{p_1+p_2}{2}\right)-\frac{H(p_1)+H(p_2)}{2} + H^\times\left(\frac{p_1+p_2}{2},(p_1p_2)_M\right)- H\left(\frac{p_1+p_2}{2}\right),\\
&=& H^\times\left(\frac{p_1+p_2}{2},(p_1p_2)_M\right) -\frac{H(p_1)+H(p_2)}{2}.
\end{eqnarray*}

Note that when $M=A$, the arithmetic mean, we have $H^\times\left(\frac{p_1+p_2}{2},(p_1p_2)_M\right)=H\left(\frac{p_1+p_2}{2}\right)$ and 
 we recover the fact that $\JS_M(p_1,p_2)=\JS(p_1,p_2)$.
\end{proof}

\section{Estimation and approximation of the extended and normalized M-JSDs}\label{sec:est}

Let us recall the two definitions of the extended M-JSD and the normalized M-JSD (for the case of $\alpha=\beta=\frac{1}{2}$)
 between two normalized densities $p_1$ and $p_2$:
\begin{eqnarray*}
\JS_M(p_1,p_2) &=& \frac{1}{2}\left( \KL\left(p_1,(p_1p_2)_M\right) + \KL\left(p_2,(p_1p_2)_M\right) \right),\\
\JS_M^+(p_1,p_2) &=& \frac{1}{2}\left(  \KL^+\left(p_1,(p_1p_2)_{\tM}\right)   +  \KL^+\left(p_2,(p_1p_2)_{\tM}\right) \right),
\end{eqnarray*}
where $(p_1p_2)_M(x)=\frac{M(p_1(x),p_2(x))}{Z_M(p_1,p_2)}$ (with $Z_M(p_1,p_2)=\int M(p_1(x),p_2(x))\,\dmu(x)$) 
and $(p_1p_2)_{\tM}(x)=M(p_1(x),p_2(x))$.

In practice, one needs to estimate  the extended and normalized G-JSDs when they do not admit closed-form formula.

\subsection{Monte Carlo estimators}

To estimate $\JS_M(p_1,p_2)$, we can use Monte Carlo samplings to estimate both KLD integrals and mixture normalizers $Z_M$;
For example, the normalizer $Z_M(p_1,p_2)$ is  estimated by
$$
\hat{Z}_M(p_1,p_2)= \frac{1}{s} \sum_{i=1}^s \frac{1}{r(x_i)}\, M(p_1(x_i),p_2(x_i)),
$$
where $r(x)$ is the proposal distribution which can be chosen according to the mean $M$ and the types of probability distributions $p_1$ and $p_2$, and $x_1,\ldots, x_s$ are $s$ identically and independently samples (iid.) from $r(x)$.
However, since $(p_1p_2)_M(x)$ is now estimated as $(p_1p_2)_{\hat{M}(x)}$, it is not anymore a normalized $M$-mixture, and thus we consider estimating
$$
\JS_{\hat{M}}^+(p_1,p_2) = \frac{1}{2}\left( \KL^+\left(p_1,(p_1p_2)_{\hat{M}}\right) + \KL^+\left(p_2,(p_1p_2)_{\hat{M}}\right) \right)
$$
to ensure the non-negativity of the divergence $\JSD_{\hat{M}}$.

Let us consider the estimation of the term
$$
\KL^+\left(p_1,(p_1p_2)_{\tM}\right) = \int \left( p_1\log\frac{p_1}{M(p_1,p_2)} + M(p_1,p_2) -p_1 \right)\,\dmu.
$$

By choosing the proposal distribution $r(x)=p_1(x)$, we have  $\KL^+\left(p_1,(p_1p_2)_{\hat{M}}\right) \approx \widehat{\KL^+}\left(p_1,(p_1p_2)_{\tM}\right)$ (for large enough $s$) where
$$
\widehat{\KL}^+\left(p_1,(p_1p_2)_{\tM}\right) = \frac{1}{s}  \sum_{i=1}^s \left(\log\frac{p_1(x_i)}{M(p_1(x_i),p_2(x_i))} + \frac{1}{p_1(x_i)}\, M(p_1(x_i),p_2(x_i)) -1\right).
$$

Monte Carlo (MC) stochastic integration~\cite{rubinstein2016simulation} is a well-studied topic in Statistics with many results on consistency and variance of MC estimators.

Note that even if we have a generic formula for the G-JSD between two densities of an exponential family given by Corollary~\ref{cor:gjsdef}, the cumulant function  $F(\theta)$ may not be in closed form~\cite{Cobb-1983,hayakawa2016estimation}.
This is the case when the sufficient statistic vector of the exponential family is $t(x)=(x,x^2,\ldots, x^m)$ (for $m\geq 5$) yielding the polynomial exponential family (also called exponential-polynomial family~\cite{hayakawa2016estimation}).

\subsection{Approximations via $\gamma$-divergences}\label{sec:gamma}

One way to circumvent the lack of computational tractable density normalizers is to consider the family of $\gamma$-divergences~\cite{fujisawa2008robust} instead of the KLD:
$$
\tD_\gamma(q_1,q_2)=\frac{1}{\gamma(1+\gamma)}\log I_\gamma(q_1,q_2) -\frac{1}{\gamma}\log I_\gamma(q_1,q_2)+\frac{1}{1+\gamma}\log I_\gamma(q_1,q_2),\quad\gamma>0,
$$
where
$$
I_\gamma(q_1,q_2) = \int q_1(x)\, q_2^\gamma(x) \,\dmu(x).
$$

The $\gamma$-divergences are projective divergences, i.e., they enjoy the property that
$$
\tD_\gamma(\lambda_1 q_1,\lambda_2 q_2)=\tD_\gamma(q_1,q_2),\quad\forall \lambda_1>0,\lambda_2>0.
$$
Furthermore, we have $\lim_{\gamma\rightarrow 0} \tD_\gamma(p_1,p_2)=\KL(p_1,p_2)$.
(Note that KLD is not projective.)

Let us define the projective $M$-JSD:
\begin{equation}\label{eq:jsdgamma}
\JS_{\tM,\gamma}(p_1,p_2) = \frac{1}{2}\left( \tD_\gamma\left(p_1,(p_1p_2)_{\tM}\right) + \tD_\gamma\left(p_2,(p_1p_2)_{\tM}\right) \right).
\end{equation}

We have for $\gamma=\eps$ small enough (e.g., $\eps\leq 10^{-3}$), $\JS_{M}(p_1,p_2)\approx \JS_{\tM,\gamma}(p_1,p_2)$ since
$$
\KL(p_1,(p_1p_2)_M) \approx_{\gamma=\eps} \tD_\gamma(p_1,(p_1p_2)_{\tM}).
$$

In particular, for exponential family densities $p_{\theta_1}(x)=\frac{q_{\theta_1}(x)}{Z(\theta_1)}$ and
$p_{\theta_2}(x)=\frac{q_{\theta_2}(x)}{Z(\theta_2)}$, we have
$$
I_\gamma(p_{\theta_1},p_{\theta_2})=\exp\left(F(\theta_1+\gamma\theta_2)-F(\theta_1)-\gamma F(\theta_2)\right),
$$
provided that $\theta_1+\gamma\theta_2$ belongs to the natural parameter space (otherwise, the integral $I_\gamma$ diverges).

Even when $F(\theta)$ is not known in closed form, we may estimate the $\gamma$-divergence by estimating the $I_\gamma$ integrals as follows:
$$
\hat{I}_\gamma(q_{\theta_1},q_{\theta_2}) \approx \frac{1}{s} \sum_{i=1}^s q_2(x_i),
$$
where $x_1,\ldots, x_s$ are iid. sampled from $p_1(x)$.
For example, we may use Monte Carlo importance sampling methods~\cite{kloek1978bayesian} or  exponential family Langevin dynamics~\cite{banerjee2022stability} to sample densities of exponential family densities with computationally intractable normalizers (e.g., polynomial exponential families).

\section{Summary and concluding remarks}\label{sec:concl}

In this paper, we first recalled the Jensen--Shannon symmetrization (JS-symmetrization) scheme of~\cite{GJSD-2019} for an arbitrary statistical dissimilarity $D(\cdot,\cdot)$ using an arbitrary weighted scalar mean $M_\alpha$ as follows:

$$
D^\JS_{M_\alpha,\beta}(p_1,p_2) \eqdef \beta \, D\left(p_1,{(p_1p_2)}_{M_\alpha}\right) + (1-\beta)\, D\left(p_2,{(p_1p_2)}_{M_\alpha}\right),\quad 
(\alpha,\beta) \in (0,1)^2, 
$$

In particular, we showed that the skewed Bhattacharyya distance and the Chernoff information can both be interpreted as JS-symmetrizations of the reverse Kullback--Leibler divergence.

Then we  defined two types of geometric Jensen--Shannon divergence between probability densities:
The first type $\JS_M$ requires to normalize $M$-mixtures and relies on the Kullback--Leibler divergence: $\JS_M=\KL^\JS_{M_{\frac{1}{2}},\frac{1}{2}}$.
The second type $\JS^+_{\tilde{M}}$ does not normalize $M$-mixtures and uses the extended Kullback--Leibler divergence $\KL^+$ to take into account unnormalized mixtures:
$\JS^+_{\tilde{M}}=\KL^{\JS^+}_{\tM_{\frac{1}{2}},\frac{1}{2}}$.
When $M$ is the arithmetic mean $A$, both $M$-JSD types coincide with the ordinary Jensen--Shannon divergence of Eq.~\ref{eq:jsd}.

We have shown that  both $M$-JSD types   can be interpreted as regularized Jensen--Shannon divergences $\JS$ with additive terms.
Namely, we have:
\begin{eqnarray*}
\JS_M(p_1,p_2)               &=& {\JS(p_1,p_2)} + \KL((p_1p_2)_A,(p_1p_2)_M),\\
\JS^+_{\tilde{M}}(p_1,p_2) &=& \JS_M(p_1,p_2) + Z_M(p_1,p_2)-\log Z_M(p_1,p_2)-1,\\
                               &=& { \JS(p_1,p_2)} + \KL((p_1p_2)_A,(p_1p_2)_M) + Z_M(p_1,p_2)-\log Z_M(p_1,p_2)-1,
\end{eqnarray*}
where $Z_M(p_1,p_2)=\int M(p_1,p_2)\,\dmu$ is the $M$-mixture normalizer.
The gap between these two types of M-JSD is
\begin{eqnarray*}
\Delta_M(p_1,p_2) &=& \JS_{\tM}^+(p_1,p_2)-\JS_M(p_1,p_2),\\
&=& Z_M(p_1,p_2)-\log Z_M(p_1,p_2)-1.
\end{eqnarray*}

When taking the geometric mean $M=G$, we showed that both G-JSD types can be expressed using the Jeffreys divergence and the Bhattacharyya divergence (or Bhattacharyya coefficient):

\begin{eqnarray*}
\JS_G(p_1,p_2) &=& \frac{1}{4} \, J(p_1,p_2)-B(p_1,p_2),\\
  \JS_{\tG}^+(p_1,p_2) &=& \frac{1}{4} \, J(p_1,p_2) +\exp(-B(p_1,p_2))-1,\\
	&=&  \frac{1}{4} \, J(p_1,p_2) +\BC(p_1,p_2)-1.
\end{eqnarray*}
Thus the gap between these two types of G-JSD is
\begin{eqnarray*}
\Delta_G(p_1,p_2) & \eqdef & \JS_{\tG}^+(p_1,p_2)-\JS_G(p_1,p_2),\\
 &=& \BC(p_1,p_2)+B(p_1,p_2)-1,\\
&=& Z_G(p_1,p_2)-\log Z_G(p_1,p_2)-1,
\end{eqnarray*}
since $Z_G(p_1,p_2)=\int \sqrt{p_1\, p_2}\dmu=\BC(p_1,p_2)$.

Although the square root of the Jensen--Shannon divergence yields a metric distance, this is not anymore the case 
for the geometric-JSD and the extended geometric-JSD: We reported counterexamples  in Remark~\ref{rk:gjsdnotmetric}.
Moreover, we have shown that the KL symmetrization $\sqrt{\KL((p_1p_2)_A,(p_1p_2)_G)}$ is not a metric distance (Remark~\ref{rk:klmn}).

We discussed the merit of the extended G-JSD which does not require to normalize the geometric mixture in \S\ref{sec:est}, 
and showed how to approximate the G-JSD using the projective $\gamma$-divergences~\cite{fujisawa2008robust} for $\gamma=\eps$, a small enough value (i.e., $\gamma=\eps=10^{-3}$). 
From the viewpoint of information geometry, the extended G-JSD has been shown to be a $f$-divergence~\cite{IG-2016} (separable divergence) while the G-JSD is not separable in general because of the normalization of mixtures (with exception of the ordinary JSD which is a $f$-divergence because the arithmetic mixtures do not require normalization).

We studied power JSDs by considering the power means and study in the $\pm\infty$ limits, the extended max-JSD and min-JSD:
We proved that the extended max-JSD is upper bounded by the total variation distance $\TV(p_1,p_2)=\frac{1}{2}\, \int |p_1-p_2|\,\dmu$:
$$
0\leq \JS^+_{\widetilde{\max}}(p_1,p_2)\leq \TV(p_1,p_2),
$$
and that the extended min-JSD is lower bounded as follows:
$$
\JS^+_{\widetilde{\min}}(p_1,p_2)\geq  \frac{1}{4}\, J(p_1,p_2)-\TV(p_1,p_2),
$$
where $J$ denotes Jeffreys's divergence: $J(p_1,p_2)=\KL(p_1,p_2)+\KL(p_2,p_1)$.

The advantage of using the extended G-JSD is that we do not need to normalize geometric mixtures while 
this novel divergence is proven to be a $f$-divergence~\cite{IG-2016} and retains the property that it amounts to a regularization of the ordinary Jensen--Shannon divergence by an extra additive gap term.
 
Finally, we expressed $\JS_G$  (Eq.~\ref{eq:GJSD-EF}) and $\JS^+_{\tilde{G}}$ (Eq.~\ref{eq:ExtGJSD-EF}) for exponential families, characterized the gap between these two types of divergences as a function of the cumulant and partition functions, and reported  
corresponding explicit formula for the multivariate Gaussian (exponential) family.  
 The G-JSD between Gaussian distributions has already been used successfully in many applications~\cite{deasy2020constraining,kumari2023rds,ni2023learning,sachdeva2024uncertainty,wang2023np,serra2024computation,thiagarajan2025jensen,hanselmann2025emperror}.

\bibliographystyle{plain}
\bibliography{YAJSSymmetrizationRevisedBIB}

\begin{thebibliography}{10}

\bibitem{AliSilvey-1966}
Syed~Mumtaz Ali and Samuel~D Silvey.
\newblock A general class of coefficients of divergence of one distribution
  from another.
\newblock {\em Journal of the Royal Statistical Society: Series B
  (Methodological)}, 28(1):131--142, 1966.

\bibitem{amari2007integration}
Shun-ichi Amari.
\newblock Integration of stochastic models by minimizing $\alpha$-divergence.
\newblock {\em Neural computation}, 19(10):2780--2796, 2007.

\bibitem{IG-2016}
Shun-ichi Amari.
\newblock {\em Information Geometry and Its Applications}.
\newblock Applied Mathematical Sciences. Springer Japan, 2016.

\bibitem{asadi2018mixture}
Majid Asadi, Nader Ebrahimi, Omid Kharazmi, and Ehsan~S Soofi.
\newblock {Mixture models, Bayes Fisher information, and divergence measures}.
\newblock {\em IEEE Transactions on Information Theory}, 65(4):2316--2321,
  2018.

\bibitem{banerjee2022stability}
Arindam Banerjee, Tiancong Chen, Xinyan Li, and Yingxue Zhou.
\newblock {Stability based generalization bounds for exponential family
  Langevin dynamics}.
\newblock In {\em International Conference on Machine Learning}, pages
  1412--1449. PMLR, 2022.

\bibitem{barndorff2014information}
Ole Barndorff-Nielsen.
\newblock {\em Information and exponential families: in statistical theory}.
\newblock John Wiley \& Sons, 2014.

\bibitem{bhattacharyya1946measure}
Anil Bhattacharyya.
\newblock On a measure of divergence between two multinomial populations.
\newblock {\em Sankhy{\=a}: the indian journal of statistics}, pages 401--406,
  1946.

\bibitem{briet2009properties}
Jop Bri{\"e}t and Peter Harremo{\"e}s.
\newblock {Properties of classical and quantum Jensen-Shannon divergence}.
\newblock {\em Physical review A}, 79(5):052311, 2009.

\bibitem{bullen2013handbook}
Peter~S Bullen.
\newblock {\em Handbook of means and their inequalities}, volume 560.
\newblock Springer Science \& Business Media, 2013.

\bibitem{cena2007exponential}
Alberto Cena and Giovanni Pistone.
\newblock Exponential statistical manifold.
\newblock {\em Annals of the Institute of Statistical Mathematics},
  59(1):27--56, 2007.

\bibitem{Cobb-1983}
Loren Cobb, Peter Koppstein, and Neng~Hsin Chen.
\newblock Estimation and moment recursion relations for multimodal
  distributions of the exponential family.
\newblock {\em Journal of the American Statistical Association},
  78(381):124--130, 1983.

\bibitem{cover1999elements}
Thomas~M Cover.
\newblock {\em Elements of information theory}.
\newblock John Wiley \& Sons, 1999.

\bibitem{Csiszar-1967}
Imre Csisz{\'a}r.
\newblock Information-type measures of difference of probability distributions
  and indirect observation.
\newblock {\em studia scientiarum Mathematicarum Hungarica}, 2:229--318, 1967.

\bibitem{csiszar2004information}
Imre Csisz{\'a}r, Paul~C Shields, et~al.
\newblock {Information theory and statistics: A tutorial}.
\newblock {\em Foundations and Trends{\textregistered} in Communications and
  Information Theory}, 1(4):417--528, 2004.

\bibitem{deasy2021alpha}
Jacob Deasy, Tom~Andrew McIver, Nikola Simidjievski, and Pietro Lio.
\newblock {$\alpha$-VAEs: Optimising variational inference by learning
  data-dependent divergence skew}.
\newblock In {\em ICML Workshop on Invertible Neural Networks, Normalizing
  Flows, and Explicit Likelihood Models}, 2021.

\bibitem{deasy2020constraining}
Jacob Deasy, Nikola Simidjievski, and Pietro Li{\`o}.
\newblock {Constraining variational inference with geometric Jensen-Shannon
  divergence}.
\newblock {\em Advances in Neural Information Processing Systems},
  33:10647--10658, 2020.

\bibitem{endres2003new}
Dominik~Maria Endres and Johannes~E Schindelin.
\newblock A new metric for probability distributions.
\newblock {\em IEEE Transactions on Information theory}, 49(7):1858--1860,
  2003.

\bibitem{fuglede2004jensen}
Bent Fuglede and Flemming Topsoe.
\newblock {Jensen-Shannon divergence and Hilbert space embedding}.
\newblock In {\em International symposium on Information theory (ISIT)},
  page~31. IEEE, 2004.

\bibitem{fujisawa2008robust}
Hironori Fujisawa and Shinto Eguchi.
\newblock Robust parameter estimation with a small bias against heavy
  contamination.
\newblock {\em Journal of Multivariate Analysis}, 99(9):2053--2081, 2008.

\bibitem{goodfellow2014generative}
Ian~J Goodfellow, Jean Pouget-Abadie, Mehdi Mirza, Bing Xu, David Warde-Farley,
  Sherjil Ozair, Aaron Courville, and Yoshua Bengio.
\newblock Generative adversarial nets.
\newblock {\em Advances in neural information processing systems}, 27, 2014.

\bibitem{grosse2013annealing}
Roger~B Grosse, Chris~J Maddison, and Russ~R Salakhutdinov.
\newblock Annealing between distributions by averaging moments.
\newblock {\em Advances in Neural Information Processing Systems}, 26, 2013.

\bibitem{grunwald2007minimum}
Peter~D Gr{\"u}nwald.
\newblock {\em The minimum description length principle}.
\newblock MIT press, 2007.

\bibitem{hanselmann2025emperror}
Niklas Hanselmann, Simon Doll, Marius Cordts, Hendrik~PA Lensch, and Andreas
  Geiger.
\newblock {EMPERROR: A Flexible Generative Perception Error Model for Probing
  Self-Driving Planners}.
\newblock {\em IEEE Robotics and Automation Letters}, 2025.

\bibitem{hayakawa2016estimation}
Jumpei Hayakawa and Akimichi Takemura.
\newblock Estimation of exponential-polynomial distribution by holonomic
  gradient descent.
\newblock {\em Communications in Statistics-Theory and Methods},
  45(23):6860--6882, 2016.

\bibitem{jeffreys1998theory}
Harold Jeffreys.
\newblock {\em The theory of probability}.
\newblock OuP Oxford, 1998.

\bibitem{jerfel2021variational}
Ghassen Jerfel, Serena Wang, Clara Wong-Fannjiang, Katherine~A Heller, Yian Ma,
  and Michael~I Jordan.
\newblock {Variational refinement for importance sampling using the forward
  Kullback-Leibler divergence}.
\newblock In {\em Uncertainty in Artificial Intelligence}, pages 1819--1829.
  PMLR, 2021.

\bibitem{johnson2001symmetrizing}
Don~H Johnson and Sinan Sinanovic.
\newblock {Symmetrizing the Kullback-Leibler distance}.
\newblock {\em IEEE Transactions on Information Theory}, 1(1):1--10, 2001.

\bibitem{jones2002general}
Lee~K Jones and Charles~L Byrne.
\newblock General entropy criteria for inverse problems, with applications to
  data compression, pattern classification, and cluster analysis.
\newblock {\em IEEE Transactions on Information Theory}, 36(1):23--30, 2002.

\bibitem{kailath1967divergence}
Thomas Kailath.
\newblock {The Divergence and Bhattacharyya Distance Measures in Signal
  Selection}.
\newblock {\em IEEE Transaction on Communication Technology}, 15:52--60, 1967.

\bibitem{kloek1978bayesian}
Teun Kloek and Herman~K Van~Dijk.
\newblock {Bayesian estimates of equation system parameters: an application of
  integration by Monte Carlo}.
\newblock {\em Econometrica: Journal of the Econometric Society}, pages 1--19,
  1978.

\bibitem{kumari2023rds}
Jeeval Kumari, Gerard Deepak, and A~Santhanavijayan.
\newblock {RDS: related document search for economics data using ontologies and
  hybrid semantics}.
\newblock In {\em International Conference on Data Analytics and Insights},
  pages 691--702. Springer, 2023.

\bibitem{JSD-1991}
Jianhua Lin.
\newblock {Divergence measures based on the Shannon entropy}.
\newblock {\em IEEE Transactions on Information theory}, 37(1):145--151, 1991.

\bibitem{melville2005active}
Prem Melville, Stewart~M Yang, Maytal Saar-Tsechansky, and Raymond Mooney.
\newblock {Active learning for probability estimation using Jensen-Shannon
  divergence}.
\newblock In {\em European conference on machine learning}, pages 268--279.
  Springer, 2005.

\bibitem{michalowicz2008calculation}
Joseph~V Michalowicz, Jonathan~M Nichols, and Frank Bucholtz.
\newblock {Calculation of differential entropy for a mixed Gaussian
  distribution}.
\newblock {\em Entropy}, 10(3):200, 2008.

\bibitem{ni2023learning}
Shuyan Ni, Cunbao Lin, Haining Wang, Yang Li, Yurong Liao, and Na~Li.
\newblock {Learning geometric Jensen-Shannon divergence for tiny object
  detection in remote sensing images}.
\newblock {\em Frontiers in Neurorobotics}, 17:1273251, 2023.

\bibitem{nielsen2014generalized}
Frank Nielsen.
\newblock {Generalized Bhattacharyya and Chernoff upper bounds on Bayes error
  using quasi-arithmetic means}.
\newblock {\em Pattern Recognition Letters}, 42:25--34, 2014.

\bibitem{GJSD-2019}
Frank Nielsen.
\newblock {On the Jensen--Shannon symmetrization of distances relying on
  abstract means}.
\newblock {\em Entropy}, 21(5):485, 2019.

\bibitem{nielsen2020elementary}
Frank Nielsen.
\newblock An elementary introduction to information geometry.
\newblock {\em Entropy}, 22(10):1100, 2020.

\bibitem{nielsen2020generalization}
Frank Nielsen.
\newblock {On a generalization of the Jensen--Shannon divergence and the
  Jensen--Shannon centroid}.
\newblock {\em Entropy}, 22(2):221, 2020.

\bibitem{nielsen2022many}
Frank Nielsen.
\newblock The many faces of information geometry.
\newblock {\em Not. Am. Math. Soc}, 69(1):36--45, 2022.

\bibitem{nielsen2022revisiting}
Frank Nielsen.
\newblock {Revisiting Chernoff information with likelihood ratio exponential
  families}.
\newblock {\em Entropy}, 24(10):1400, 2022.

\bibitem{nielsen2025two}
Frank Nielsen.
\newblock {Two Types of Geometric Jensen--Shannon Divergences}.
\newblock {\em Entropy}, 27:947, 2025.

\bibitem{nielsen2011burbea}
Frank Nielsen and Sylvain Boltz.
\newblock {The Burbea-Rao and Bhattacharyya centroids}.
\newblock {\em IEEE Transactions on Information Theory}, 57(8):5455--5466,
  2011.

\bibitem{nishimura2008information}
Tomoaki Nishimura and Fumiyasu Komaki.
\newblock The information geometric structure of generalized empirical
  likelihood estimators.
\newblock {\em Communications in Statistics—Theory and Methods},
  37(12):1867--1879, 2008.

\bibitem{okamura2023metrization}
Kazuki Okamura.
\newblock {Metrization of powers of the Jensen-Shannon divergence}.
\newblock {\em arXiv preprint arXiv:2302.10070}, 2023.

\bibitem{osterreicher2003new}
Ferdinand Osterreicher and Igor Vajda.
\newblock A new class of metric divergences on probability spaces and its
  applicability in statistics.
\newblock {\em Annals of the Institute of Statistical Mathematics},
  55(3):639--653, 2003.

\bibitem{rubinstein2016simulation}
Reuven~Y Rubinstein and Dirk~P Kroese.
\newblock {\em {Simulation and the Monte Carlo method}}.
\newblock John Wiley \& Sons, 2016.

\bibitem{sachdeva2024uncertainty}
Rewat Sachdeva, Raghav Gakhar, Sharad Awasthi, Kavinder Singh, Ashutosh Pandey,
  and Anil~Singh Parihar.
\newblock {Uncertainty and Noise Aware Decision Making for Autonomous Vehicles
  - A Bayesian Approach}.
\newblock {\em IEEE Transactions on Vehicular Technology}, 2024.

\bibitem{schoenberg1938metric}
Isaac~J Schoenberg.
\newblock Metric spaces and completely monotone functions.
\newblock {\em Annals of Mathematics}, 39(4):811--841, 1938.

\bibitem{serra2024computation}
Giuseppe Serra, Photios~A Stavrou, and Marios Kountouris.
\newblock {On the computation of the Gaussian rate--distortion--perception
  function}.
\newblock {\em IEEE Journal on Selected Areas in Information Theory},
  5:314--330, 2024.

\bibitem{sibson1969information}
Robin Sibson.
\newblock Information radius.
\newblock {\em Zeitschrift f{\"u}r Wahrscheinlichkeitstheorie und verwandte
  Gebiete}, 14(2):149--160, 1969.

\bibitem{sutter2020multimodal}
Thomas Sutter, Imant Daunhawer, and Julia Vogt.
\newblock {Multimodal generative learning utilizing Jensen-Shannon-divergence}.
\newblock {\em Advances in neural information processing systems},
  33:6100--6110, 2020.

\bibitem{taneja1995new}
Inder~Jeet Taneja.
\newblock New developments in generalized information measures.
\newblock In {\em Advances in imaging and electron physics}, volume~91, pages
  37--135. Elsevier, 1995.

\bibitem{thiagarajan2025jensen}
Ponkrshnan Thiagarajan and Susanta Ghosh.
\newblock {Jensen--Shannon divergence based novel loss functions for Bayesian
  neural networks}.
\newblock {\em Neurocomputing}, 618:129115, 2025.

\bibitem{virosztek2021metric}
D{\'a}niel Virosztek.
\newblock {The metric property of the quantum Jensen-Shannon divergence}.
\newblock {\em Advances in Mathematics}, 380:107595, 2021.

\bibitem{wang2023np}
Jianfeng Wang, Daniela Massiceti, Xiaolin Hu, Vladimir Pavlovic, and Thomas
  Lukasiewicz.
\newblock {NP-SemiSeg: when neural processes meet semi-supervised semantic
  segmentation}.
\newblock In {\em International Conference on Machine Learning}, pages
  36138--36156. PMLR, 2023.

\bibitem{yamano2019some}
Takuya Yamano.
\newblock {Some bounds for skewed $\alpha$-Jensen-Shannon divergence}.
\newblock {\em Results in Applied Mathematics}, 3:100064, 2019.

\end{thebibliography}

\appendix

\section{Notations}\label{sec:notations}

{\small
\begin{supertabular}{ll}
Means:\\
$M_\alpha(a,b)$ & weighted scalar mean\\
$M_\alpha^\phi(a,b)$ & weighted quasi-arithmetic scalar mean for generator $\phi(u)$\\
$A(a,b)$ & arithmetic mean\\
$A_\alpha(a,b)$ & weighted arithmetic mean\\
$G_\alpha(a,b)$ & weighted geometric mean\\
$G(a,b)$ & geometric mean\\
$P_\gamma(a,b)$ & power mean with $P_0=G$ and $P_1=A$\\
$P_{\gamma,\alpha}(a,b)$ & weighted power mean\\
Densities on measure space $(\calX,\calE,\mu)$:\\
$p, p_1, p_2, \ldots$ & normalized density\\
$q, q_1, q_2, \ldots$ & unnormalized density\\
$Z(q)$ & density normalizer $p=\frac{q}{Z(q)}$\\
$Z_M(p_1,p_2)$ & normalizer of $M$-mixture ($\alpha=\frac{1}{2}$)\\
$\hat{Z}_M(p_1,p_2)$ & Monte Carlo estimator of $Z_M(p_1,p_2)$\\
$Z_{M,\alpha}(p_1,p_2)$ & normalizer of weighted $M$-mixture\\
$(p_1p_2)_M$ & $M$-mixture \\
$(p_1p_2)_{M,\alpha}$ & weighted $M$-mixture\\
Dissimilarities, divergences, and distances:\\
$\KL(p_1,p_2)$ & Kullback--Leibler divergence (KLD)\\
$\KL^+(q_1,q_2)$ & extended Kullback--Leibler divergence\\
$\KL^*(p_1,p_2)$ & reverse Kullback--Leibler divergence\\
$H^\times(p_1,p_2)$ & cross-entropy\\
$H(p)$ & Shannon discrete or differential entropy\\
$J(p_1,p_2)$ & Jeffreys divergence\\
$\TV(p_1,p_2)$ & total variation distance\\
$B(p_1,p_2)$ & Bhattacharyya ``distance'' (not metric)\\
$B_\alpha(p_1,p_2)$ & $\alpha$-skewed Bhattacharyya ``distance''\\
$C(p_1,p_2)$ & Chernoff information or Chernoff distance\\
$T(p_1,p_2)$ & Taneja $T$-divergence\\
$I_f(p_1,p_2)$ & Ali-Silvey-Csisz\'ar $f$-divergence\\
$D(p_1,p_2)$ & arbitrary dissimilarity measure\\
$D^*(p_1,p_2)$ & reverse dissimilarity measure\\
$D^+(q_1,q_2)$ & extended dissimilarity measure\\
$\tD(q_1,q_2)$ & projective dissimilarity measure\\
$\tD_\gamma(q_1,q_2)$ & $\gamma$-divergence\\
$\hat{D}^+(q_1,q_2)$ & Monte Carlo estimation of dissimilarity $D^+$\\
Jensen--Shannon divergences and generalizations:\\
$\JS(p_1,p_2)$ & Jensen--Shannon divergence (JSD)\\
$\JS_{\alpha,\beta}(p_1,p_2)$ & $\beta$-weighted $\alpha$-skewed mixture JSD\\
$\JS_M(p_1,p_2)$ &  $M$-JSD for $M$-mixtures\\
$\JS_G(p_1,p_2)$ &  geometric JSD\\
$\JS_{\tG}(p_1,p_2)$ & extended  geometric JSD\\
$\JS_G^*(p_1,p_2)$ & left-sided  geometric JSD (right-sided for $\KL^*$)\\ 
${\JS^+}_{\widetilde{\min}}(p_1,p_2)$ & min-JSD\\
${\JS^+}_{\widetilde{\max}}(p_1,p_2)$ & max-JSD\\
$\Delta_M(p_1,p_2)$ & gap between extended and normalized M-JSDs\\
\end{supertabular}
}
\end{document}